\documentclass[acmsmall,screen,nonacm]{acmart}
\AtBeginDocument{%
  }

\setcopyright{none}
\copyrightyear{2026}
\acmYear{2026}
\acmDOI{XXXXXXX.XXXXXXX}

\acmJournal{PACMPL}
\acmVolume{0}
\acmNumber{0}
\acmArticle{0}
\acmMonth{0}

\usepackage{cleveref}
\usepackage{minted}
\usepackage{mathtools}
\usepackage{outlines}
\usepackage{tikz}
\usepackage{amsmath}
\usepackage{bm}
\usepackage{dirtytalk}
\usepackage[titletoc]{appendix}
\usepackage{subcaption}

\usepackage{tikz}
\usetikzlibrary{shapes.geometric, positioning}

\setminted[haskell]{breaklines,fontsize=\small}

\definecolor{dgreen}{rgb}{0,0.4,0}
\definecolor{dblue}{rgb}{0,0,0.7}

\newif\ifarxiv
\arxivtrue   % arXiv version
\newif\ifcomments
\commentsfalse       % comments OFF: original layout, notes vanish

\newlength{\notewidth}
\newlength{\notesep}  
\newlength{\noteedge} 

\ifcomments
  \usepackage{geometry}
  \newlength{\saveleft}  
  \newlength{\savetextw} 
  \newlength{\savepaperw}
  \newlength{\saveright} 
  \newlength{\notetotal} 
  \newlength{\newpaperw} 
  \newlength{\newright}  
  \makeatletter \@mparswitchfalse \makeatother
  \normalmarginpar
  \usepackage[textsize=footnotesize, colorinlistoftodos]{todonotes}
  \usepackage{marginfix}            % shifts bottom-anchored notes UP into free margin space
  \makeatletter
  \def\dumpmargins{%
    \@tempcnta\z@
    \loop
    \unless\ifx\mfx@marginlist\@empty
      \let\temp@\mfx@marginlist
      \vbox{}\clearpage
      \ifx\temp@\mfx@marginlist            % stuck: nothing got placed this pass
        \advance\@tempcnta\@ne
        \ifnum\@tempcnta>5\relax
          \PackageWarning{marginfix}{Could not place some margin notes;
            dropping them. Use \string\AD[inline]{...} for long comments}%
          \global\let\mfx@marginlist\@empty
        \else
          \extendmargin{2\textheight}\vbox{}\clearpage   % emergency: force them out
        \fi
      \fi
    \repeat
  }
  \makeatother
\else
  \usepackage[disable]{todonotes}   % geometry untouched -> original format
\fi

\newcommand{\takeaway}[1]{\smallskip\par\noindent%
  \textbf{Consequences for PBT:} \emph{#1}\par\smallskip}

\newcommand{\GEN}{\textit{GEN}}
\newcommand{\RE}{\textit{RE}}

\newcommand{\NTIME}{\textit{NTIME}}
\newcommand{\NSPACE}{\textit{NSPACE}}
\newcommand{\GTIME}{\textit{GTIME}}
\newcommand{\GSPACE}{\textit{GSPACE}}
\newcommand{\DTIME}{\textit{DTIME}}
\newcommand{\DSPACE}{\textit{DSPACE}}

\newcommand{\Poly}{\textit{P}}
\newcommand{\NP}{\textit{NP}}
\newcommand{\PG}{\textit{PG}}
\newcommand{\EPG}{\textit{EPG}}
\newcommand{\LG}{\textit{LG}}
\newcommand{\NL}{\textit{NL}}
\newcommand{\REACH}{\textit{REACH}}

\newcommand{\BPP}{\textit{BPP}}
\newcommand{\PSPACE}{\textit{PSPACE}}
\newcommand{\PSPACEG}{\textit{PSPACEG}}
\newcommand{\NPSPACE}{\textit{NPSPACE}}
\newcommand{\EXPTIME}{\textit{EXPTIME}}

\newcommand{\PARTIAL}{\textit{PARTIAL}}
\newcommand{\DEFAULT}{\textit{DEFAULT}}

\newcommand{\Lang}[1]{\mathcal{L}(#1)}
\newcommand{\logr}{\operatorname{log}}

\newtheorem{theorem}{Theorem}[section]
\newtheorem{definition}{Definition}[section]

\newtheorem{lemma}{Lemma}[section]
\newtheorem{proposition}{Proposition}[section]
\newtheorem{corollary}{Corollary}[section]

\newtheorem{observation}{Observation}[section]

\newtheorem{example}{Example}[section]

\newcommand{\blanksymbol}{\,
  \tikz[baseline=-0.4ex]{ % lower baseline to sit under text
    \draw[line width=0.4pt]
      (0,0) -- (0,-0.4ex)           % left vertical downward
      -- (0.8em,-0.4ex)             % bottom horizontal
      -- (0.8em,0);                 % right vertical upward
  }\,
}

\newcommand{\concat}{\mathbin{+\!\!+}}

\newsavebox{\boxA}
\newsavebox{\boxB}
\newsavebox{\boxC}
\newlength{\snippetwd}
\newlength{\snippetht}
\theoremstyle{acmdefinition}
\newtheorem{remark}{Remark}

\begin{document}

\title{Complexity Theory of Randomised Testing}    

\author{Pingshi Yu}
\email{p.yu22@imperial.ac.uk}
\orcid{0000-0002-4998-4878}
\affiliation{%
  \institution{Imperial College London}
  \city{London}
  \country{United Kingdom}
}

\author{Chengsong Tan}
\email{tanchengsong@kaihong.com}
\orcid{0009-0008-7822-8407}
\affiliation{%
  \institution{Kaihong}
  \city{Shenzhen}
  \country{China}
}

\author{Nicolas Wu}
\email{n.wu@imperial.ac.uk}
\orcid{0000-0002-4161-985X}
\affiliation{%
  \institution{Imperial College London}
  \city{London}
  \country{United Kingdom}
}

\author{Alastair Donaldson}
\email{alastair.donaldson@imperial.ac.uk}
\orcid{0000-0002-7448-7961}
\affiliation{%
  \institution{Imperial College London}
  \city{London}
  \country{United Kingdom}
}

\begin{abstract}
Randomised testing is a widely-used approach to software validation,
yet despite years of practical development its theoretical foundations remain thin. 
In particular, the fundamental question of what it means for a set of inputs to be \emph{generable} has gone unanswered in both the literature and folklore.
We present, for the first time, complexity-theoretic foundations for random generators in software testing.
We model generators as Turing machine transducers that consume random bits and produce string-encoded outputs,
and show that the theoretically generable languages coincide exactly with the recursively enumerable languages.
This has direct implications for testing at the boundaries of decidability, such as in the field of compiler testing.
Turning to \emph{efficient} generation,
we show that the polynomial-time generable languages lie within \NP{},
that certain important \NP{}-complete languages admit efficient generators,
and that---under standard cryptographic assumptions---there are languages in \Poly{} for which no efficient generator exists: the complexity of efficienct generation and of efficient decision are not the same.
We then show that space-bounded complexity is the natural framework for generators producing \emph{correlated} samples,
capturing methodologies such as coverage-guided fuzzing and symbolic execution.
Beyond classification, we characterise efficient generability: a language has a polynomial-time generator iff it admits a \emph{certificate scheme} over a verifier---so witness planting, the folklore technique behind generators to test SAT solvers, is in a sense the only route to efficient generation.
Our theory also yields design principles for property-based testing libraries: we prove no library can compositionally derive efficient generators from logical predicates involving conjunction or negation, under standard assumptions. 
However, restricted classes like \NL{} (equivalently, linear Datalog predicates) would admit such a compilation.

\end{abstract}

\maketitle

\section{Introduction}\label{sec:intro}

At first sight, \emph{software testing},
concerned with the observed behaviour of particular programs on particular
inputs,
is far removed from \emph{complexity theory},
which distinguishes what can be computed in principle from what can be computed by an efficient algorithm.
Yet, in \emph{property-based testing} (PBT)~\cite{claessen2000quickcheck, maciver2019hypothesis}, where a \emph{system under test} (SUT) is executed on many automatically generated inputs,
a critical component to the success of a testing campaign is the \emph{generator}, a randomised algorithm. 
Due to the intricacies of realistic generators, much work has gone into their implementation~\citep{yang2011finding,
livinskii2020random},
library tooling~\citep{maciver2019hypothesis,
claessen2000quickcheck, yu2025ratte, hodovan2018grammarinator},
and on automatic \emph{derivation} of generators from predicates~\citep{goldstein2022parsing, lampropoulos2017generating, goldstein2026search}.
But despite their importance, little is known about the
\emph{theoretical limits} of generator expressivity:
what can be generated \emph{at all}, and what can be generated \emph{efficiently}?

\begin{figure}
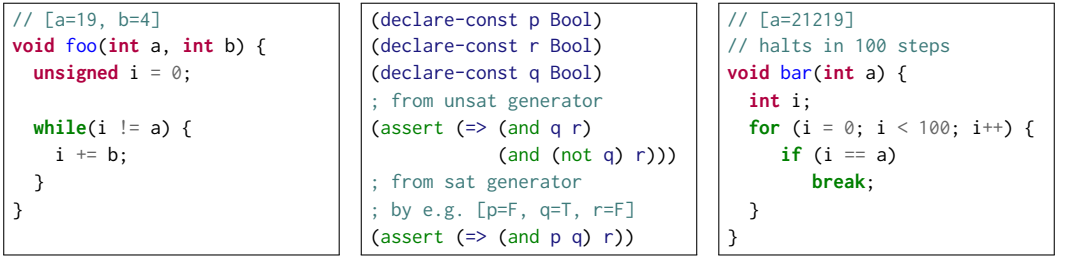

  \centering
  \setminted{fontsize=\footnotesize, frame=none, framesep=0pt, baselinestretch=1.05}
  \setlength{\fboxsep}{3pt}
  \setlength{\snippetwd}{\dimexpr0.32\columnwidth-2\fboxsep-2\fboxrule\relax}

  % 1. Capture each snippet — NOTE the [t] : top-aligned, not centered -------
  \begin{lrbox}{\boxA}\begin{minipage}[t]{\snippetwd}
\begin{minted}{C}
// [a=19, b=4]
void foo(int a, int b) {
  unsigned i = 0;
  
  while(i != a) {
    i += b;
  }
}
\end{minted}
  \end{minipage}\end{lrbox}
  \begin{lrbox}{\boxB}\begin{minipage}[t]{\snippetwd}
\begin{minted}{lisp}
(declare-const p Bool)
(declare-const r Bool)
(declare-const q Bool)
; from unsat generator
(assert (=> (and q r)
            (and (not q) r)))
; from sat generator 
; by e.g. [p=F, q=T, r=F]
(assert (=> (and p q) r))
\end{minted}
  \end{minipage}\end{lrbox}

  \begin{lrbox}{\boxC}\begin{minipage}[t]{\snippetwd}
\begin{minted}{C}
// [a=21219] 
// halts in 100 steps
void bar(int a) {
  int i;
  for (i = 0; i < 100; i++) {
     if (i == a) 
        break;
  }
}
\end{minted}
  \end{minipage}\end{lrbox}

  % 2. Tallest of the three -------------------------------------------------
  \setlength{\snippetht}{\dimexpr\ht\boxA+\dp\boxA\relax}
  \ifdim\dimexpr\ht\boxB+\dp\boxB\relax>\snippetht
    \setlength{\snippetht}{\dimexpr\ht\boxB+\dp\boxB\relax}\fi
  \ifdim\dimexpr\ht\boxC+\dp\boxC\relax>\snippetht
    \setlength{\snippetht}{\dimexpr\ht\boxC+\dp\boxC\relax}\fi

  % 3. Emit framed at common height, each box top-aligned (outer [t]) --------
  \begin{subfigure}[t]{0.32\columnwidth}
    \fbox{\begin{minipage}[t][\snippetht][t]{\snippetwd}\usebox{\boxA}\end{minipage}}
    \smallskip\par
    \caption{Nonterminating C program.}
    \label{fig:example-testcase-non-term}
  \end{subfigure}
  \hfill
  \begin{subfigure}[t]{0.32\columnwidth}
    \fbox{\begin{minipage}[t][\snippetht][t]{\snippetwd}\usebox{\boxB}\end{minipage}}
    \smallskip\par
    \caption{(Un)satisfiable formulas.}
    \label{fig:example-testcase-sat-unsat}
  \end{subfigure}
  \hfill
  \begin{subfigure}[t]{0.32\columnwidth}
    \fbox{\begin{minipage}[t][\snippetht][t]{\snippetwd}\usebox{\boxC}\end{minipage}}
    \smallskip\par
    \caption{Terminating C program.}
    \label{fig:example-testcase-bounded-term}
  \end{subfigure}

  \caption{Possible test cases produced by example generators.}
  \label{fig:example-testcase}
\end{figure}

We give the first complexity-theoretic formalisation of generators in randomised testing,
establishing a number of results about what can and cannot be generated,
and why.
Our formulation is applicable to the majority of generators found in testing practice:
generators are modelled as Turing machine transducers that consume bitstrings and must \emph{eventually} produce string-encoded outputs.
Once generators are cast within this framework,
each generator associates with a formal language---the set of outputs it can produce.
This allows rigorous investigation of the limits of generation in general, as well as under a variety of time and space constraints.
In establishing our novel framework, we build upon previous theoretical work on language generation that predates PBT (and which in our view deserves more attention)~\citep{sanchis1990efficient, sanchis1990complexity}.

To illustrate the kinds of questions our theoretical framework allows us to ask and answer, we give three examples in the context of programming languages research.

\smallskip\noindent\textbf{Example: testing a termination checker.}
Suppose we are presented with a termination checking tool for C-programs, which is claimed to be conservative:
the tool should emit a ``may not terminate'' diagnostic for \emph{at least} those $(\text{program}, \text{input})$ pairs where the program does not terminate on the input.
To apply PBT to this tool we thus need a generator of elements from the set $\mathbf{NonTerm} = \{ (p, i) : p \text{ a C program, } i \textrm{ an input, }p \textrm{ does not terminate on } i\}$.
If testing finds a $(p, i)$ pair for which the checker does \emph{not} emit a diagnostic, it has found a bug in the checker.
\Cref{fig:example-testcase-non-term} illustrates one such pair that should be producible by the generator.
While of course we can create a generator that produces many such pairs of examples (possibly including \Cref{fig:example-testcase-non-term}), our theory tells us that it is \emph{impossible} to create such a generator without sacrificing some part of the $\mathbf{NonTerm}$ space:
we show that the recursively enumerable (\RE{}) languages are precisely the generable ones, and since $\mathbf{NonTerm}$ is not in \RE{} (assuming an unbounded memory semantics for C), it cannot be generable.

\smallskip\noindent\textbf{Example: testing a SAT solver.}
To apply PBT to a SAT solver we would like two generators,
generating satisfiable (unsatisfiable) formulas: if the solver claims that any are unsatisfiable (satisfiable) this is a bug.
\Cref{fig:example-testcase-sat-unsat} uses SMT-LIB syntax to sketch simple formulas that such generators might be expected to produce.
Our theory confirms that generators covering the full space of satisfiable or unsatisfiable problems can \emph{exist},
and that a generator of satisfiable problems that runs in polynomial time can be constructed (even though deciding satisfiability is NP-complete).
However, we show that (unless $\Poly{}=\NP{}$) a fully general polynomial-time generator of unsatisfiable problems cannot exist.
Furthermore, if we move from SAT to QBF instances by allowing quantifiers, a fully general polynomial-time generator cannot exist even for the satisfiable case.

\smallskip\noindent\textbf{Example: time-bounded differential compiler testing.}
As a final motivating example, a campaign for randomised differential compiler testing typically requires $(\text{program}, \text{input})$ pairs where the program \emph{does} terminate on the input.
To make such a campaign efficient we might further want the generator to provide a time bound for each generated program, avoiding the need for arbitrarily-selected timeout values.
We would thus like to generate $(p, i, k)$ triples comprising a program $p$ that is guaranteed to terminate on input $i$ within $k$ time steps\footnote{Assuming that each basic operation of the program takes one time step.} (see \Cref{fig:example-testcase-bounded-term}).
Although this input space \emph{is} generable,
we show that unless standard complexity-theoretic assumptions turn out to be false,
a time-efficient generator \emph{cannot} cover the set of all such programs: no generator of independent samples exist that uses $q(n)$ time to generate programs of size $n$,
for some polynomial $q$.
Not only that, polynomial generation remains out of reach even if the generator had access to an efficient SAT solver.
The results established in \Cref{sec:space-generators} show that polynomial space $q(n)$ would also not be sufficient to produce all such programs of size $n$---limiting the use of feedback-driven generation techniques such as concolic execution~\citep{cadar2021klee, busse2020running} and coverage-guided fuzzing~\citep{bohme2016coverage} for this particular set.
A practical solution would have to make some engineering trade-offs.

Beyond classification, our framework yields a \emph{characterisation} of efficient generability: we show that a language has a polynomial-time generator if and only if it has a verifier admitting a \emph{certificate scheme}---a way to sample certificates and reconstruct random instances around them.
This shows that \emph{some form} of witness planting (including one based on a trivial verifier)---the folklore technique behind generators for SAT and graph reachability---is used by \emph{all} time-efficient generators.

Separately, a long-standing goal in the PBT literature is to automatically derive efficient generators from predicates~\cite{goldstein2026search, claessen2014generating, lampropoulos2017generating}. 
We prove a negative result in this context: under standard assumptions, no library can \emph{compositionally} map predicates from a logic containing conjunction or negation to polynomial-time generators---any homomorphic translation must give up either expressivity or efficiency.
On the other hand, we show that restricted classes like \NL{} (equivalently, predicates in linear Datalog) can in principle be compiled compositionally to polynomial-time generators.

By studying the power of generators under various resource constraints, we establish a wide range of results about testing that fall out naturally. 
Some confirm intuitions from testing folklore; others are novel insights that were previously inaccessible without such a formal model.

\smallskip\noindent\textbf{Contributions.}
In summary, our main contributions are:
\begin{itemize}
    \item We provide the first complexity-theoretic foundations of generators for randomised testing (\Cref{sec:generators}), and derive expressivity bounds applicable to all generators (\Cref{sec:general-statements}).
    \item We propose a notion of time/space constraints for generators, and derive theoretical limits of efficient generation (\Cref{sec:generator-families}, \Cref{sec:polytime-generators}, \Cref{sec:space-generators}).
    \item We establish general principles for constructing time-efficient generators, and limitations on the design of composable testing libraries (\Cref{sec:practitioners}).
\end{itemize}

Throughout the paper, we introduce several generator-specific complexity classes: $\LG{}$ (log-space generable), $\PG{}$ (polynomial-time generable), $\EPG{}$ (expected polynomial-time generable),
$\PSPACEG{}$ (polynomial-space generable) and $\GEN{}$ (everything that can be generated).
\Cref{fig:complexity-classes} serves as a roadmap to the complexity landscape established in our work.

\iffalse 
\textbf{Paper organisation.}
This paper is organised as follows.
We formalise the notion of generators in \Cref{sec:generators}.
In \Cref{sec:general-statements} we investigate unconstrained generators and relate them to recursively enumerable languages.
In \Cref{sec:generator-families} we discuss and define time and space constraints for generators.
The notion of polynomial-time generators are defined and investigated in \Cref{sec:polytime-generators}, and in \Cref{subsec:pg-np} we give an example for a language likely not polynomial time generable.
An alternative class of expected polynomial time generators are considered in \Cref{subsec:epg}.
We explore space-constrained generators in \Cref{sec:space-generators}, along with their relations to \NP{} and time-bounded generation classes.
We discuss some implications of our results for practitioners and library designers in \Cref{sec:practitioners}.
In \Cref{sec:related-work} we discuss related work, and conclude in \Cref{sec:conclusion}.
\fi

%%%%%%%%%%%%%%%%%

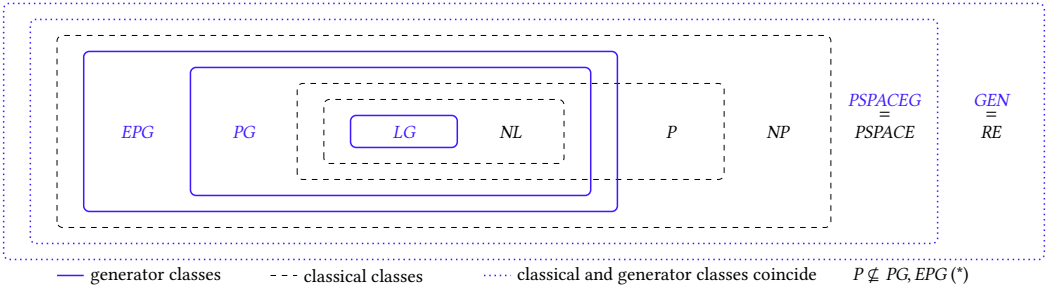
\begin{figure}[t]
\centering
\resizebox{\linewidth}{!}{% Preamble requirements:
% \usepackage{tikz}
% \usepackage{xcolor}

\definecolor{gencolor}{HTML}{3C19FF}

\begin{tikzpicture}[x=0.5cm, y=0.15cm]

\tikzset{
  gensolid/.style  = {draw=gencolor, thick,  solid, rounded corners=3pt},
  gendashed/.style = {draw=gencolor, thick, dotted, rounded corners=3pt},
  classical/.style = {draw=black,     thin, dashed, rounded corners=3pt},
}

% GEN = RE (outermost, purple dashed)
\draw[gendashed] (0, 0) rectangle (39, 32);
\node[color=gencolor] at (37, 20) {$\GEN$};
\node at (37, 18) {$=$};
\node at (37, 16) {$\RE$};

% PSPACEG = PSPACE (purple dashed)
\draw[gendashed] (1, 2) rectangle (35, 30);
\node[color=gencolor] at (33, 20) {$\PSPACEG{}$};
\node at (33, 18) {$=$};
\node at (33, 16) {$\PSPACE{}$};

% NP (classical, grey dashed)
\draw[classical] (2, 4) rectangle (31, 28);
\node at (29, 16) {$\NP{}$};

% P (classical, grey dashed)
\draw[classical] (11, 10) rectangle (27, 22);
\node at (25, 16) {$\textit{P}$};

% NL (classical, grey dashed)
\draw[classical] (12,12) rectangle (21, 20);
\node at (19, 16) {$\NL{}$};

% EPG (generator, purple solid)
\draw[gensolid]  (3, 6) rectangle (23, 26);
\node[color=gencolor] at (5, 16) {$\EPG{}$};

% PG (generator, purple solid)
\draw[gensolid]  (7, 8) rectangle (22, 24);
\node[color=gencolor] at (9, 16) {$\PG{}$};

% LG (generator, purple solid)
\draw[gensolid]  (13, 14) rectangle (17, 18);
\node[color=gencolor] at (15, 16) {$\LG{}$};

% ---- Legend ----

\draw[gensolid]  (2, -2) -- (3, -2);
\node[anchor=west] at (3, -2) {generator classes};

\draw[classical] (10, -2) -- (11, -2);
\node[anchor=west] at (11, -2) {classical classes};

\draw[gendashed] (18, -2) -- (19, -2);
\node[anchor=west] at (19, -2) {classical and generator classes coincide};

\node at (34, -2) {$\textit{P} \not\subseteq \PG{}, \EPG{}$ (*)};

\end{tikzpicture}}
\caption{Complexity relationships between the new generators classes and known classes.
Purple dotted boundaries denote generator classes that do not coincide with classical ones; purple solid boundaries denote generator classes that do coincide with classical ones; black dashed boundaries denote classical complexity classes that do not coincide with generator classes.
The generator classes are \LG{}: logspace generable, \PG{}: polynomial-time generable, \EPG{}: expected polynomial-time generable, \PSPACEG{}: polynomial-space generable, and \GEN{}: all generable languages.
$\textit{P}$ contains $\NL{}$, intersects $\PG{}$ and
$\EPG{}$, and is contained in $\NP{}$.
(*) Non-containment of $\textit{P}$ is contingent on standard cryptographic assumptions.}
\label{fig:complexity-classes}
\end{figure}

\section{Modelling Generators}\label{sec:generators}

A \emph{generator} in randomised testing produces structured inputs for an SUT~\citep{mili2015software},
e.g.\ inputs that the SUT is supposed to accept, or deliberately malformed inputs designed to stress the SUT.
These structured inputs can be viewed as strings in a language.
For generality, and to allow the application of complexity theory, we formulate generators as Turing machines, and their outputs as string encodings of data.
We model randomness as an explicit input bitstring supplied to a deterministic Turing machine, where
random selections are made during computation by reading from this bitstring.
This is in line with how randomness is implemented in common randomised testing libraries such as Hypothesis~\citep{maciver2019hypothesis}.\footnote{%
It is also possible to model randomness using nondeterministic Turing machines, which would result in an equivalent model.
We choose deterministic Turing machines with an explicit random input to stay in line with practical implementations, and because it makes later definitions more convenient.
The results for the deterministic case translate directly to the nondeterministic case.}

In practice, generation is driven by a \emph{pseudorandom number generator} (PRNG), which we idealise as an inexhaustible supply of independent, unbiased random bits: a generator consumes some finite prefix of an infinite stream and produces an output. 
There are two caveats with the informal description, which we make formal below. 
First, we permit a generator that produces nothing for \emph{every} stream, characterising generators of the empty set, and is useful for our theoretical model. 
Second, a generator being productive ``given an infinite stream'' is too strong if interpreted as ``given \emph{every} infinite stream''.
We show in \Cref{sec:eventual-productivity} that reasonable generators can still fail to produce a value in extreme, albeit probability 0 events---so we require producivity only \emph{almost surely}.

Although it is natural to define generators to take infinite streams as input, infinite streams are not a natural input space for Turing machines.
To reconcile this, we proceed in three steps:

\begin{enumerate}
  \item We define \emph{bitstring transducers}, which operate on \emph{finite} bitstrings, together with a \emph{halting} condition.
    A halting transducer either accepts, producing an output, or rejects, where a rejection only occurs when the transducer runs out of its finite input (\Cref{sec:finite-transducers}).
    % This property is also present in real generator implementations.

  \item We lift a halting transducer to infinite streams.
    Because rejection only occurs after input exhaustion, acceptance is monotone under extension of the input, thus every accepting run belongs to an equivalence class represented by a unique minimal accepting prefix.
    This lets us define what a transducer does on an infinite stream, and lets us \emph{measure} the set of accepted bitstreams (\Cref{sec:infinite-streams}).
    
  \item We define the notion of \emph{eventual productivity} as almost-sure acceptance of infinite streams under the measure of infinite sequences of fair coin-flips, and obtain the definition of a generator (\Cref{sec:eventual-productivity}).\footnote{%
    In \Cref{appendix:alternative-eventually-productive} we discuss an alternative natural notion of eventual productivity based only on Turing machines.}
\end{enumerate}

\paragraph{Preliminaries}
Unless stated otherwise, all tapes of a \emph{Turing machine} (TM) are assumed to be infinite (to the right), initially empty, with the head pointed at some fixed, designated initial cell.
The alphabet on all tapes of Turing machines is finite, and includes the ``$\blanksymbol$'' symbol to represent blank spaces.
The blank symbol will never be written.
When we make the statement ``tape with alphabet $\Gamma$'', we implicitly mean it has alphabet $\Gamma \uplus \{\blanksymbol\}$, and all empty cells on the tape are populated by $\blanksymbol$ symbols.
We say a tape is read-once-only if the head never moves left and the machine never writes to the tape.
Similarly, a tape is write-once-only if the head never moves left, and after any write operation the tape head is immediately moved to the right.
For a Turing machine $M$ with a single input tape, we say ``$M$ ran on $x$'', or ``$M$ ran with input $x$'', or ``$M(x)$'', to mean the computation performed by $M$ in the initial state, with the input tape populated using $x$ starting from the initial cell.
This also extends to machines with multiple input tapes.
For example, on a machine $M$ with two input tapes, $M(x, y)$ means the computation performed by $M$ in the initial state, with the first input tape populated using $x$ and the second using $y$.

\subsection{Transducers on finite bitstrings}\label{sec:finite-transducers}

\begin{definition}[Bitstring Transducer]\label{def:transducer}
  A \emph{bitstring transducer} $T$ is a Turing machine with three tapes: a read-once-only \emph{input tape} with alphabet $\{0, 1\}$, a \emph{working tape} with alphabet $\Sigma$, and a write-once-only \emph{output tape} with alphabet $\Gamma$.
  For a bitstring $b \in \{0, 1\}^*$, if $T(b)$ accepts, the contents of the output tape (disregarding empty cells) $x \in \Gamma^*$ are said to be \emph{produced} by $T$, denoted $T(b) = x$.

  If there is no ambiguity, we may say transducer for short.
\end{definition}

\begin{definition}[Random Selection]\label{def:rand-selection}
  During an execution, we say that the transducer $T$ makes a \emph{random selection} whenever a transition moves the input tape head right---that is, it reads a bit from the input bitstring.
\end{definition}

Multiple random bits can be combined together in a \emph{procedure} to make more refined random selections (e.g.\ conditional dependence, or selections among more than two possibilities).
Procedures are computations (potentially with side effects) that consume randomness in the sense of \Cref{def:rand-selection}, and generators are special cases of procedures that write values to the output tape.
Procedures are used to modularise algorithmic descriptions and help with clarity of exposition.

\begin{example}\label{def:naturalencoding}
  A transducer $N$ can select any $n \in \mathbb{N}$ at random via the following iterative procedure.
  $N$ reads the input in pairs of bits.
  In each round it reads a \emph{continuation bit}; if it is $1$, $N$ reads a further \emph{data bit}, appends it to a buffer on the working tape, and begins a new round; if it is $0$, $N$ halts and outputs the buffer, interpreted as a binary numeral.
  Writing the input stream as $b_0 b_1 b_2 \ldots$, the transducer computes
  $\sum_{i=0}^{k-1} b_{2i+1} 2^{i}$ where $k = \min\{ i : b_{2i} = 0 \}$, whenever this $k$ exists.
  Note that the only way $N$ rejects is when it runs out of input bits, e.g.\ on $1$ or $101$ (the continuation bit promises another bit that is not supplied), or on $10$ (no continuation bit is supplied).
\end{example}

The accepting runs of \Cref{def:naturalencoding} consume $2k+1$ bits for some $k\geq 0$,
and rejections occur due to the stream ending prematurely.
This is the behaviour we now generalise to our \emph{halting} condition.

\begin{definition}[Halting Transducer]\label{def:halting-transducer}
  A bitstring transducer $T$ is \emph{halting} if
  \begin{enumerate}
    \item for every input $b \in \{0, 1\}^*$, $T(b)$ halts; and
    \item $T$ rejects iff the input head attempts to read a blank cell (a cell containing $\blanksymbol$).
  \end{enumerate}
\end{definition}

The second clause prevents premature rejections and is a necessary precursor for the formal definition of ``eventual productivity'' below.
From \Cref{def:halting-transducer}, we can immediately observe the following, useful for the upcoming generalisation to infinite streams.

\begin{observation}[Monotonicity]\label{obs:monotone}
  Let $T$ be halting and suppose $T(b)$ accepts.
  Then the run $T(b)$ reads no cell beyond the $|b|$-th, since by \Cref{def:halting-transducer} reading a blank would force rejection.
  Hence for every $b'$ having $b$ as a prefix, the run of $T$ on $b'$ is identical, so $T(b')$ accepts and $T(b') = T(b)$.
\end{observation}

\subsection{From finite bitstrings to infinite streams}\label{sec:infinite-streams}

By \Cref{obs:monotone}, the accepting behaviour of a halting transducer $T$ is completely determined by the shortest bitstrings on which it accepts.

\begin{definition}\label{def:accepting-minimal-bitstings}
  For a halting transducer $T$, the \emph{minimal accepted bitstrings} $A_T$ are
  $$
    A_T = \{\;b \in \{0, 1\}^* : T(b) \text{ accepts}\ \land\ \forall b'.\ b'\ \text{a strict prefix of}\ b \implies T(b')\ \mathrm{rejects} \;\}.
  $$
\end{definition}

Again by \Cref{obs:monotone}, no element of $A_T$ is a prefix of another, i.e.\ $A_T$ is prefix-free, and $T$ accepts a finite bitstring $b$ exactly when some element of $A_T$ is a prefix of $b$.
Note that $A_T$ is defined for conceptual convenience, and is not intended to be computed directly.

We can now assign a natural semantics of $T$ on an infinite bitstream $\textit{bs} \in \{0,1\}^\omega$: its behaviour is identical to $T(b)$, where $b \in A_T$ is a prefix of $bs$, if $b$ exists.

\begin{definition}\label{def:stream-semantics}
  For a halting transducer $T$ and an infinite bitstream $\textit{bs}$, define $T(\textit{bs})$ as:
  \begin{enumerate}
    \item $T(b)$, if $b \in A_T$ and $b$ is a prefix of $\textit{bs}$;
    \item divergent, if no $b \in A_T$ is a prefix of $\textit{bs}$.
  \end{enumerate}
  This is well-defined: as $A_T$ is prefix-free, at most one $b \in A_T$ is a prefix of $\textit{bs}$.
\end{definition}

To speak of the probability that $T$ produces an output, we use the standard measure of infinite sequences of fair coin-flips, which models the PRNG that drives generators in practice.

\begin{definition}[Distribution of Random Bitstreams]
  Let $\mathcal{I}$ denote the distribution of bitstreams $\{0, 1\}^\omega$, where for a random variable $X$ distributed according to $\mathcal{I}$ (written $X \sim \mathcal{I}$), all bits of $X$ are independently and identically distributed (i.i.d.) following a fair Bernoulli distribution $\mathrm{Ber}(0.5)$.
  For a set of bitstreams $B \subseteq \{0, 1\}^\omega$, denote the probability of $B$ under $\mathcal{I}$ by $\mu(B)$.
\end{definition}

\begin{definition}
  For a bitstring $b \in \{0, 1\}^*$, define the \emph{cylinder} $S_b$ as:%  the set of all infinite bitstreams with prefix $b$:
  $$S_b = \{b \concat \textit{bs} : \textit{bs} \in \{0, 1\}^\omega\}.$$
\end{definition}

\begin{definition}\label{def:accepted-streams}
  The \emph{accepted} infinite bitstreams of a halting transducer $T$ are
  $$A_T^\omega = \bigcup_{b \in A_T} S_b.$$
\end{definition}

Since $A_T$ is prefix-free, the cylinders $S_b$ for $b \in A_T$ are pairwise disjoint.
$A_T$ is also countable, thus $A_T^\omega$ is a countable disjoint union of cylinders, and is measurable with
$$\mu(A_T^\omega) = \sum_{b \in A_T} 2^{-|b|}.$$
This quantity is the probability that $T$, when driven by an idealised PRNG, produces an output.

\subsection{Eventual productivity}\label{sec:eventual-productivity}

We are now ready to formalise the ``eventual productivity'' requirement of a generator.
The strongest interpretation of eventual productivity is that $T(\textit{bs})$ is defined for \emph{every} $\textit{bs} \in \{0,1\}^\omega$, equivalently $A_T^\omega = \{0,1\}^\omega$.
However, this is too strong in practice.
Consider the natural number transducer $N$ of \Cref{def:naturalencoding}.
It diverges on exactly the streams in:
$$
X = \{ b_0b_1\ldots : \forall i \in \mathbb{N}.\ i \textrm{ even} \Rightarrow b_i = 1 \},
$$
that is, streams in which every continuation bit is $1$, so that $N$ never stops asking for another data bit.
Ruling out transducers like $N$ would make it impossible to sample from an infinite set---any transducer with infinitely many possible outputs must have infinitely many accepting prefixes, and so there must exist streams on which the transducer is never productive.
This would severely limit the applicability of our model.

Observe however, that the set $X$ has trivial measure, as $\mu(X) = \lim_{i \rightarrow \infty}2^{-i} = 0$.
We therefore formalise \emph{eventual productivity} as \emph{almost-sure} acceptance.

\begin{definition}[Generators]\label{def:generators}
  A halting transducer $G$ is a \emph{generator} if $\mu(A_G^\omega) = 1$ or $\mu(A_G^\omega) = 0$.
\end{definition}

\begin{observation}
  $\mu(A_G^\omega) = 0$ iff $A_G = \emptyset$ iff $G$ rejects every finite bitstring.
\end{observation}
\begin{proof}
  Each $S_b$ has $\mu(S_b) = 2^{-|b|} > 0$, so $\mu(A_G^\omega) = \sum_{b \in A_G} 2^{-|b|} = 0$ iff $A_G$ is empty.
  By \Cref{obs:monotone}, $A_G = \emptyset$ iff $G$ accepts no finite bitstring.
\end{proof}

The two cases ($\mu(A_G^\omega) = 1 \text{ or } 0$) are for the two behaviours permitted of generators: the first characterises almost-sure productivity; the second describes the case where all inputs are rejected (corresponding to a generator for the empty set).
It is undecidable to check whether a given transducer is a generator, but this is not an issue since we show generator conditions on a case-by-case basis throughout this work.\footnote{In a similar vein, the undecidability of the Halting problem does not prevent people from developing useful algorithms.}
\ifarxiv
For added intuition, it may help to consider the transducers that are \emph{not} generators.
These are the transducers that fail to halt on some finite bitstring, or that diverge on a non-negligible --- but not co-negligible --- set of streams.
\fi
To illustrate the definition, we verify that it matches our intuitions on \Cref{def:naturalencoding}, and that rejection sampling yields true generators.

\begin{proposition}
  The transducer $N$ of natural numbers described in \Cref{def:naturalencoding} is a generator.
\end{proposition}
\begin{proof}
  $N$ is clearly halting: on any finite bitstring it either meets a $0$ continuation bit and accepts, or exhausts its input and rejects.
  The set of bitstreams \emph{not} accepted by $N$ is $X$ as above, so $A_N^\omega = \{0,1\}^\omega \setminus X$ and $\mu(A_N^\omega) = 1 - \mu(X)$.
  Since $\mu(X) = 0$, $\mu(A_N^\omega) = 1$ and $N$ is a generator.
  
  % Let $X_k = \{ \textit{bs} : b_{2i} = 1 \text{ for all } i < k \}$.
  % The $X_k$ are decreasing with $X = \bigcap_{k} X_k$, and as the bits are i.i.d.\ we have $\mu(X_k) = 2^{-k}$.
  % By continuity from above, $\mu(X) = \lim_{k \to \infty} 2^{-k} = 0$, so $\mu(A_N^\omega) = 1$ and $N$ is a generator.
\end{proof}

\begin{definition}[Generator by Rejection Sampling]
  Let $f$ be a total computable predicate on $\{0, 1\}^*$.
  The generator $G_f$ based on rejection sampling on $f$ is defined as follows.
  \begin{enumerate}
    \item Sample an $n$ (e.g.\ using \Cref{def:naturalencoding}), then sample a string $w$ of length $n$.
    \item Compute $f(w)$.
      If $f(w)$ returns true, output $w$.
      Otherwise restart from step 1.
  \end{enumerate}
\end{definition}

\begin{proposition}\label{prop:rejection-sampling-valid}
  If $f$ is a computable predicate on $\{0, 1\}^*$ such that $f$ holds for at least one $w \in \{0, 1\}^*$, then $G_f$ is a generator.
\end{proposition}
\begin{proof}
  See \Cref{prop:termination-proof}.
\end{proof}

\section{A Framework for Generability}\label{sec:general-statements}

In testing, randomness is an essential tool to enable systematic exploration of the input space. 
The goal of randomised testing can be viewed as the validation of a certain property that must hold on the program when executed on the random inputs. 
The commonly practised paradigm of \emph{property-based testing} (PBT) validates properties of the form $\forall x. \mathit{Pre}(x) \implies \mathit{Post}(x)$, which suits the vast majority of use cases.
For instance, fuzzing can be seen as validating the property: ``for all inputs $x$, the program should not crash''.
Differential testing can be seen as validating: ``for all valid inputs $x$, the two implementations $T, T'$ must have $T(x) = T'(x)$''.

The key to the successful application of randomised testing is an effective generator, whose job is to generate inputs that satisfy the preconditions of the property, which then enables the property to be checked on the program.
We now show a few example scenarios below of programs and their properties that may be tested, and will be referring back to them throughout this section.

\begin{example}[Binary Search Trees]\label{example:bst}
Program $P$ takes as input a binary search tree $T$ and an index $i$, and outputs the $i$th smallest element within the tree.
$P$ should satisfy the following property, where $\mathit{bst}(T)$ is a predicate checking that $T$ is a binary search tree, and $i, j$ are integers:
$$
  \forall T, i, j. \mathit{bst}(T) \land 0 \leq i \leq j < |T|
            \implies P(T,i) \leq P(T,j)
$$
\end{example}
\begin{example}[Counting SAT Assignments]\label{example:counting-sat}
Program $P$ takes as input a CNF boolean formula $\psi$ and outputs the number of satisfying assignments of $\psi$.
$P$ should satisfy the following property, where $\mathit{sat}$ is the predicate for the satisfiability of the formula, and $\psi[x/A]$ is the substitution of $A$ for $x$ in $\psi$:
$$
    \forall \psi. \mathit{sat}(\psi) 
        \implies P(\psi) > 0 \land P(\psi) = P(\psi[x_0/T]) + P(\psi[x_0/F])
$$
\end{example}

\begin{example}[Tensor Decomposition]\label{example:tensor-decomposition}
Program $P$ takes a tensor $T$ over field $F$, and outputs a decomposition of $T$ as a product of rank-1 tensors.
$P$ should satisfy the following:
\[ \forall T,r.\ \mathit{rank}(T)=r \Longrightarrow \mathit{length}(P(T))=r \ \wedge\ \left(\forall i<r.\ \mathit{rank}(P(T)_i)=1\right) \ \wedge\ \bigotimes_{i<r} P(T)_i = T . \]
\end{example}

\begin{example}[C Compiler]\label{example:c-compiler}
Program $P$ is a compiler for the C language, taking an optimisation flag (0 or 1) as the first argument and a C program as the second argument.
$P$ should satisfy the following property, where $\mathit{wd}(c)$ is a predicate checking $c$ is a C program with well-defined semantics (assume it doesn't take inputs), and $P(o, c)()$ denotes executing the executable produced by compiling $c$ using $P$ on optimisation $o$:
$$
    \forall c. \mathit{wd}(c)
        \implies P(0, c)() = P(1, c)()
$$
\end{example}

In each case, for a property $\mathit{Prop} = \forall x. \mathit{Pre}(x) \implies \mathit{Post}(x)$, the ideal generator should implement a distribution on the elements of the set $\mathit{Precond} = \{x : \mathit{Pre}(x)\}$.

\subsection{The Design Space of Generability}

We first address some nuances in the design space of the formal language associated with a generator.
This is an important step, since languages are the central object of study in any complexity theory. 

\paragraph{Allow generators to sample only a subset of $\mathit{Precond}$?}
Sampling from a strict subset of $\mathit{Precond}$ would mean there exist elements $x \in \mathit{Precond}$ that can never be generated. 
Such a generator can still find bugs, but never ones triggered only by inputs outside its range.
Thus the generator of a subset of $\mathit{Precond}$ would not be appropriate to \emph{fully} test $\mathit{Prop}$.

As a pragmatic choice, it is common in practice to create generators that sample from only a subset of inputs that satisfies a property's precondition.
For instance, when the precondition is seemingly too complex\footnote{Note this is only an empirical understanding, which is thus far not formalised. When is a set really \emph{too complex} to generate? We will look to give theoretical answers to this question in the coming sections.} to write a generator for, such as the property described in \Cref{example:c-compiler}, generators typically focus on specific subsets of the input space heuristically understood to contain potential bug triggers~\citep{livinskii2020random, even2022csmithedge}.
Although existing literature refers to these as ``generators'' of C-programs for validating compilers, the property they are testing for is more accurately understood as a different from the one in \Cref{example:c-compiler}, with a more restrictive precondition:
$$\forall c. \mathit{producible}(c) \implies P(0, c)() = P(1, c)()$$
where $\mathit{producible}$ is to denote the predicate corresponding to the targeted subset of inputs the generator produces.
For example, restricted C-programs amenable to implemented analyses~\citep{yang2011finding}, or C-programs containing constrained forms of loops~\citep{livinskii2023fuzzing}.

\paragraph{Allow generators to sample a superset of $\mathit{Precond}$?}

Generators sampling from a superset of $\mathit{Precond}$ would be able to trigger all possible property violations, but risk false positives from inputs outside of $\mathit{Precond}$. % have a risk of triggering false positives when an input outside of $\mathit{Precond}$ is used in testing.
Rejection sampling is required to filter the inputs, selecting only ones that satisfy the precondition.
Depending on requirements, this filter can be applied either as a final step of generation, or after a bug-triggering input has been found to check for false-positivity.

While filtering works well when deciding the precondition is easy (\Cref{example:bst}), 
the precondition itself can be hard to decide, and is not in general related to the complexity of the algorithm under test.
For example, checking the precondition is \NP-hard in \Cref{example:counting-sat} and \Cref{example:tensor-decomposition}, and undecidable in \Cref{example:c-compiler} (deciding if a C program is well-defined is undecidable~\citep{ellison2012executable}).
\Cref{example:tensor-decomposition} is worth highlighting here: computing the rank of $T$ is \NP{}-hard, but if $\mathit{rank}(T)$ is known, the decomposition is computable in time polynomial in the size of the tensor.
The runtime of the test campaign can thus be dominated by the filtering stage.
These filters also need to be provided separately, and when the filter is complex, it can be an additional source of complexity and bugs in the testing system.

Moreover, test case throughput is an important factor for the effectiveness of a testing campaign~\citep{marcozzi2019compiler, bohme2020fuzzing}.
Spending significant time deciding the validity of inputs can often mean fewer or missed bugs.
If a precondition-satisfying input is produced rarely by a generator, for instance, using a generator of random ASTs (or even random strings) for \Cref{example:c-compiler}, then it can be the case that no tests are ever run, due to the low probability of finding a satisfying input!
Indeed, if the generator produces precondition-satisfying inputs with probability $p$, then one expects $1/p$ attempts to find a usable input---and $1/p$ can easily scale exponentially with the size of the generated test case. 
In the \Cref{example:c-compiler} case, a random AST would have vanishingly low probability of satisfying the language specification~\citep{cppstandard1998}, and a random string even more so.

\paragraph{What distributions should generators implement?}

A generator implements a particular distribution over the set of inputs, and in practice impacts the effectiveness of a testing campaign. 
Randomised testing employs a wide range of distributions, which can also be adjusted as testing progresses~\cite{padhye2019semantic, groce2012swarm}. 
However, there is no canonical distribution around which to build a complexity theory.
The uniform distribution (on elements of a given size) is commonly studied theoretically, but fail to capture the diversity of distributions used in practice. 
Indeed, in the design of generators, the distribution implemented is often an orthogonal concern to samplability itself.

Therefore, as a first step towards a complexity-theoretic model of generators, we remain agnostic to distributions, and aim to gain fundamental understanding of \emph{which} sets can be generated at all.
This serves as a prerequisite to the harder question of generability under \emph{prescribed} distributions, and also ensures generality of our results.
The ``no-go'' results we derive, a type of result common in complexity theory, apply to generators implementing \emph{any} computable distribution, rather than being contingent on distributional assumptions. 
The only requirement we impose is \emph{distributional support}~\cite{paraskevopoulou2015foundational}: the support of the implemented
distribution must equal the whole set.

One natural concern is that the weak assumption of only distributional support may be overly permissive, since it allows distributions that assign negligible weight to parts of the set.
However, distributional support is a stronger guarantee than it may first appear---once an input is producible, however improbably, practical techniques can often amplify the probability of producing desirable outcomes.
For example, the bitstring input to a generator can be viewed as an optimisation parameter, and this insight was used in previous works~\citep{padhye2019semantic, goldstein2022parsing, maciver2019hypothesis} to tune generators, producing outputs that would have otherwise been infeasible due to their low probability weights.

\subsection{Unconstrained Generability}
In light of the insights from the above discussion, we give the following definition for the language associated with a generator.

\begin{definition}[Generated Language]\label{def:generated-language}
For a generator $G$, the \emph{language generated by $G$} is:
$$\Lang{G} = \{ x \mid \exists \; b \in \{0, 1\}^+ \,.\, G(b) = x \}$$
\end{definition}

The following is the most general class of languages that is generable.
The only requirement is that there \emph{exists} an algorithm (generator) to generate the elements.

\begin{definition}[Generable Languages]
  For an alphabet $\Sigma$, we define the set of \emph{generable languages}, $\GEN{}$, as follows:
  
  $$\GEN{} = \{ \Lang{G} : G \textrm{ is a generator}\}$$

  A language $L$ is \emph{generable} if $L \in \GEN{}$.
\end{definition}

For a language $L$, we have defined our generators to associate with $L$ a surjective computable function $\{0, 1\}^* \rightarrow L$ that constructs instances of $L$ from randomness.
On the other hand, one of the most commonly studied objects in complexity theory is Turing recognisers, which implement a (potentially partial) predicate for some language $L$.
Turing recognisers associate a language $L$ with a function of type $\Sigma^* \rightarrow \mathbb{B}$, which accepts $x \in \Sigma^*$ iff $x \in L$.

The standard class of $\RE{}$, or recursively enumerable languages, is the set of all languages whose elements can be recognised by a Turing machine (predicate).

\begin{definition}[Recursively Enumerable Languages]\label{def:re}
For a given alphabet $\Sigma$, the recursively enumerable languages are:

$$
\RE{} = \{ L \subseteq \Sigma^* :\;
  \exists\;\text{a Turing machine } M.\; \forall x \in \Sigma^*, \quad x \in L \Leftrightarrow M\ \text{halts and accepts } x
\}
$$  
\end{definition}

In \Cref{def:re}, $M$ refers to a standard Turing machine that accepts/rejects an input, rather than the transducer style within \Cref{def:transducer}.
$M$ can be seen as a membership predicate for some set $S$, where $M(x)$ accepts if and only if $x \in S$ (but importantly may loop instead of rejecting if $x \notin S$).
The $\RE{}$ sets/languages with a membership predicate algorithm that always halts are called \emph{decidable}, %
and the complement of the decidable languages in $\mathcal{P}(\Sigma^*)$ is called \emph{undecidable}.

The well-definedness of a C-program ($\mathit{wd}$ in \Cref{example:c-compiler}) is an undecidable problem, making it a ``hard'' member of $\RE{}$.
As noted earlier, the natural approach of rejection sampling cannot be used within any generator of well-defined C-programs (programs $x$ satisfying $\mathit{wd}(x)$), since generators must terminate on any finite input, and the test for a string $x$ satisfying $\mathit{wd}(x)$ may never terminate.
The inability to use rejection sampling of course extends to all $\RE{}$ languages.

This raises the natural question: is it actually possible to write a generator for well-defined C-programs, or indeed any undecidable set?
The question is not limited to a single generation algorithm that one may write: any conceivable procedure, or combinations of procedures whose set of outputs equals the set of well-defined C-programs would be a generator. 

We now show that in the resource unconstrained case, the generable languages are precisely the recursive enumerable languages.
This is consistent with the known fact that standard transducers have images equalling the $\RE{}$ languages---the additional constraints of a generator do not sacrifice any expressive power.
However, for testing this already leads to interesting conclusions, which we discuss in more detail later in the section: programs such as compilers often do deal with undecideability, and this places fundamental constraints on the types of properties that can be fully tested on such programs. 
Later in \Cref{sec:polytime-generators} and \Cref{sec:space-generators}, we will see that additional resource (time or space) constraints lead to situations where the generator and their corresponding decision complexity classes do not coincide.

\begin{theorem}[Generability is Recursive Enumerability]\label{prop:re-equals-gen}
$\RE{} = \GEN{}$.
\end{theorem}
\begin{proof}
We defer the full proof to \Cref{proof:re-equals-gen}.
Intuitively, $\GEN{} \subseteq \RE{}$ is due to the $\RE{}$ machine being able to enumerate through all finite bitstrings (as random seeds $b$) to check if the generator can produce some $x$.
$\RE{} \subseteq \GEN{}$ is due to the generator being able to perform dovetailing to find a single accepted string by the recogniser, which can then be used as fallback for a nondeterministic guess of the membership certificate for some randomly chosen $x$.
\end{proof}

\takeaway{When testing a property for which the set of acceptable inputs is recursively enumerable, \Cref{prop:re-equals-gen} implies that a generator for the complete input format exists---i.e., ambition in generator-writing is at least theoretically justified.
However, properties whose preconditions are not recursively enumerable cannot be fully tested by any generator, no matter how it is engineered.}

For instance, it is possible to write a generator that will output, with non-zero probability, all instances of seemingly very hard problems, such as:
\begin{definition}\label{def:halting-language}
The language of halting computations on Turing machines:
$$ \textit{HALT} = \{ \langle M, w \rangle : \textrm{Turing machine}\ M\ \text{halts on input}\ w \} $$
\end{definition}

The equality also yields a class of non-generable languages---ones that lie outside of \RE{}.
\begin{corollary}
All languages outside of $\RE{}$ are not generable.
\end{corollary}

Non-$\RE{}$ languages are a large set, and one way to obtain languages in it is via complementing languages known to be in \RE{}.
A language $L$ is undecidable if there is no Turing machine that halts on all inputs and accepts exactly the strings in $L$.
Rice's theorem~\cite{rice1953classes} is a well-known result that states all (non-trivial) semantic properties on Turing machines are undecidable.
Hence, at least one of the language itself, or its complement, is non-\RE{}.
The result also extends from Turing machines to Turing-complete programming languages, and we will quote this version.

For a Turing-complete programming language $K$, let $P_K$ be the set of
encodings of well-defined programs written in $K$, and for $M \in P_K$ let
$\Lang{M}$ be the set of strings $w$ such that $M(w)$ executes without
rejection (crashes or errors).
A \emph{semantic property} is a set $S \subseteq P_K$ that depends only on
program semantics: if $\Lang{M} = \Lang{N}$ then $M \in S \iff N \in S$.
It is \emph{non-trivial} if $S \neq \emptyset$ and $S \neq P_K$, and its
complement is $S^c = P_K \setminus S$.
Rice's theorem~\citep{rice1953classes} implies every non-trivial semantic
property is undecidable, yielding the following:

\begin{corollary}
Let $S \subseteq P_K$ be a non-trivial semantic property over programs written in a Turing complete language $K$.
Then at least one of $S$ or $S^c$ is not in $\GEN{}$.
\end{corollary}
\begin{proof}
Suppose for a contradiction both $S$ and $S^c$ are generable, then both are also recursively enumerable by \Cref{prop:re-equals-gen}, thus $S$ is actually decidable. 
This contradicts Rice's Theorem.
\end{proof}

These results are particularly applicable to software that can come face-to-face with non-\RE{} languages, such as compilers.
In this context, it means that there are properties of the compiler implementation that can never be completely tested, due to the lack of a generator for the property.
This includes the use of rejection sampling. 
For example, in C++, termination is a property relevant for the well-definedness of programs, defined by the language specification.
One might be interested in building a generator to test the compiler (its inputs being programs) for the following: programs that do not terminate should nonetheless not crash the compiler.
However, although the terminating programs are generable, there are no generators of all non-terminating programs:
\begin{corollary}\label{prop:loop-not-generable}
For a Turing-complete language $K$, the set of non-Halting programs, $\textit{LOOP}$, is not generable, where
$$\textit{LOOP} = \{ M \in P_K : M \text{ takes no inputs and does not terminate} \}$$
\end{corollary}

\takeaway{You cannot fully test compiler properties with non-\RE{} preconditions, e.g., the behaviour on non-terminating programs (\Cref{prop:loop-not-generable}).
Any generator of a producible subset is testing a \emph{different, weaker} property---a trade-off that should be made explicitly, not discovered accidentally.}

Of course, the existence of generators is only one of the concerns of PBT practitioners.
We now turn to another important consideration: their \emph{efficiency}.

\section{Generator Families}\label{sec:generator-families}
We cannot wait all day for a generator to produce an output for an SUT: an effective generator needs to rapidly produce a large corpus of inputs for the SUT. 
Thus, as with algorithms, it is helpful to associate resource constraints with generators. 
But, unlike algorithms, generators do not have a natural input parameter on which the input size can be defined, and hence there also lacks a definition of resource constraint, which are based on the input size.

In practice, PBT libraries like QuickCheck indeed include an optional \emph{size} parameter, which a generator implementation can use to control the sizes of the generated objects.

We introduce a similar parameter in our model, which allows a priori specification of the output size---from which we can then define the notion of resource usage for generators.
Adding these constraints would allow our model to be applicable to generators that use the size parameter.

The following aims to capture the notion of parameterised generators, namely a family of generators that each generates a portion of the intended language (instances of a particular size), and where the union of their generated languages is the whole language.

\ifarxiv
\begin{definition}[Implementation of a Family of TMs]
Let $(M_i)_{i \in C}$ be a family of TMs indexed by strings $C \subseteq \{0, 1\}^*$.
We say that $M$ implements $(M_i)_{i \in C}$, written $M = (M_i)_{i \in C}$, if $(M_i)_{i \in C}$ is a uniform family as follows:
\begin{itemize}
    \item $M$ has an extra read-only input tape compared to each $M_i$.
    \item For all $i \in C$, $M(i) = M_i$. 
    That is, with $i$ on the input tape of $M$, $M$ behaves exactly as $M_i$. 
\end{itemize}
\end{definition}

\begin{definition}[Family of Generators]\label{def:size-dep-generator}
Let $L$ be a language, and $(L_i)_{i \in C}$ a partition of $L$ indexed by the set $C \subseteq \{0, 1\}^*$. 
$G = (G_i)_{i \in C}$ is a generator for $L$ if for each $i \in C, \Lang{G_i} = L_i$. 

The language generated by $G$, $\Lang{G}$ is defined as
$\Lang{G} = \bigcup_{i \in \mathbb{N}} \Lang{G_i}$.
\end{definition}

For the special case where $C = \{1^i : i \in \mathbb{N}\}$ is used to partition a set $L$ into $L_i = \{w \in L : |w| = i\}$, we will also denote $C$ as $\mathbb{N}$, and $L_{1^n}$ as $L_n$, for brevity.
This will be used when we consider length-based partitions of a language $L$ in the following sections.
\else
\begin{definition}[Implementation of a Family of TMs]
Let $(M_i)_{i \in \mathbb{N}}$ be a family of TMs.
We say that $M$ implements $(M_i)_{i \in \mathbb{N}}$, written $M = (M_i)_{i \in \mathbb{N}}$, if $(M_i)_{i \in \mathbb{N}}$ is a uniform family:
\begin{itemize}
    \item $M$ has an extra read-only input tape compared to each $M_i$.
    \item For all $i \in \mathbb{N}$, $M(i) = M_i$. 
    That is, with $i$ on the input tape of $M$, $M$ behaves exactly as $M_i$. 
\end{itemize}
\end{definition}

\begin{definition}[Family of Generators]\label{def:size-dep-generator}
The family of TMs $G = (G_i)_{i \in \mathbb{N}}$ is a family of generators if, for all $i \in \mathbb{N}$, $G_i$ is a generator that outputs strings of length $i$. 
Write $\Lang{G}$ to denote the language generated by $G$, where $\Lang{G} = \bigcup_{i \in \mathbb{N}} \Lang{G_i}$.

% Let $L$ be a language, and $(L_i)_{i \in \mathbb{N}}$ a partition of $L$ where $L_i = \{ x \in L : |x| = i \}$. 
\end{definition}
\fi

Using this, we can define the generable languages under space and time constraints.

\begin{definition}\label{defn:time-bounded-generator}
For a time constructible function $t : \mathbb{N} \rightarrow \mathbb{N}$, the generator family $G = (G_i)_{i \in \mathbb{N}}$ is time bounded by $t$ if for all $i \in \mathbb{N}$, $G_i$ terminates within time $t(i)$.

Define the class of languages $\GTIME(t(n))$ generable by some $t$-time bounded generators as:
$$
\GTIME(t(n)) = \{ \Lang{G} : \forall i. G_i\text{ is time bounded by $t(i)$} \}
$$
\end{definition}

\begin{definition}\label{defn:space-bounded-generator}
For a space constructible function $s : \mathbb{N} \rightarrow \mathbb{N}$, the generator family $G = (G_i)_{i \in \mathbb{N}}$ is space bounded by $s$ if for all $i \in \mathbb{N}$, $G_i$ uses no more than $s(i)$ cells on its work tape.

Define the class of languages $\GSPACE(s(n))$ generable by some $s$-space bounded generators as:
$$
\GSPACE(s(n)) = \{ \Lang{G} : \forall i. G_i\text{ is space bounded by $s(i)$} \}
$$
\end{definition}

The following is immediate for both time and space bounded generators:
\begin{proposition}\label{prop:subset-nondet}
Suppose $s : \mathbb{N} \rightarrow \mathbb{N}$ is space-constructible with $s(i) > \logr(i)$, and $t : \mathbb{N} \rightarrow \mathbb{N}$ is time-constructible with $t(i) > i$, then:
$$
    \GSPACE(s(n)) \subseteq \NSPACE(s(n)) \qquad
    \GTIME(t(n)) \subseteq \NTIME(t(n))
$$
\end{proposition}
\begin{proof}
By simulation.
For a full proof, see \Cref{proof:subset-nondet}.
\end{proof}

\subsection{Alternative Notions of Size}
In our formalisation, a generator $G$ for a language $L$ will produce, on input $n$, a random element $x \in L$ of length \emph{exactly} $n$.
We discuss here some alternative possible formalisations of instance size, inspired by testing practice.

\subsubsection{On generating elements of size $|x| \leq n$}
Some generators in practice produce, given an $n \in \mathbb{N}$, elements of size $|x| \leq n$ rather than $|x| = n$.
These generators can fit within our framework by observing that such generator $G$ of the language $L$ can be mapped to an equivalent generator of the following language:
$$L_{\mathit{pad}} = \{ x\#0^i : x \in L \land i \in \mathbb{N} \}$$
where $0^i$ denotes $i$ repetitions of the reserved element $0$, and $\#$ is a dedicated separator element.
If $L_{\mathit{pad}}$ is generable by the $|x| = n$ scheme (for a given $n$, produce elements of length exactly $n$), then $L$ is generable by the $|x| \leq n$ scheme (for a given $n$, produce elements of length at most $n$).
Moreover, observe that the decision complexity of $L_{pad}$ is the same as $L$, hence the later no-go results from connections to decision complexity classes also applies for $|x| \leq n$ generation schemes.

\subsubsection{Other parameterisations of size}
Apart from string length, there can also be other concievable notions of size for a language $L$, in forms of some computable function $sz : L \rightarrow \mathbb{N}$.
For instance, for CNF formulas, the number of clauses or variables can both be useful notions of size.

However, it is difficult to define notions of complexity on metrics like the clauses or variables count.
This is because formulas with a fixed number of clauses or variables can, in general, describe formulas of arbitrary size in terms of their string representations.
Grouping these together in a complexity-theoretic formalisation would mean assigning the same measure of complexity to all such elements---yet the arbitrary lengths of their string representations would genuinely take vastly different (particularly time) resources to produce.

For $sz$ based on the number of clauses or variables, the string representations of $x$ with $sz(x) = m$ has unbounded ratios: $\forall n \in \mathbb{N}. \exists m \in \mathbb{N}. \exists x \in L_m. \frac{|x|}{m} > n$, where $L_m  = \{x \in \Sigma^* : sz(x) = m\}$.
Alternatively, consider another $sz$ where $sz(x)$ is within some range of $|x|$, e.g.
$$(1-\delta)|x|\leq sz(x) \leq (1 + \delta)|x|$$
In these cases, we can recover the fixed-length paradigm by padding all elements with $sz(x) = m$ to a string of length exactly $(1 + \delta)m$: the $sz$-$m$ elements are generable in the $sz$-based paradigm if and only if length $(1+\delta)m$ elements are generable in the fixed-length paradigm.

\section{Polynomial-Time Generators}\label{sec:polytime-generators}
In this section, we discuss generators that run in polynomial time, which corresponds to the class of realistically implementable generators.

We note that prior work~\citep{sanchis1990complexity, sanchis1990efficient} also considered polynomial-time generators, and actually predates PBT~\citep{claessen2000quickcheck} itself.
In the course of establishing a number of new results, we will also walk through some previously known results (\Cref{prop:pg-np}, \Cref{prop:cnf-sat-pg}, \Cref{prop:pg-oracle-superset-np}) within our framework to highlight the links between modern-day PBT and complexity theory.

An equivalent version of the following definition was first introduced in Definition 3.1 of~\citep{sanchis1990complexity}.
\begin{definition}[Polynomial-Time Generators]\label{def:size-poly-generator}
The generator family $G = (G_i)_{i \in \mathbb{N}}$ runs in polynomial time if there exists a polynomial $p$ such that for all $i \in \mathbb{N}$, $G_i$ terminates in time $p(i)$. 
\end{definition}

\begin{remark}
We do not need the empty set clause of Definition 3.1 in~\citep{sanchis1990complexity}, because we define the generator of an empty set as one that rejects all inputs, and we can use its rejection within $p(i)$ steps (on a length $p(i)$ input bitstring) to determine the emptiness of $\Lang{G_i}$.
On the other hand, if $\Lang{G_i} \neq \emptyset$, then it can never reject when provided with a length $p(i)$ input bitstring.
\end{remark}

\begin{definition}
The class of languages generable by some polynomial-time generator is:
$$ \PG{} = \bigcup_{i \in \mathbb{N}} \GTIME(n^i) $$
\end{definition}

Generators in randomised testing often use \emph{rejection sampling} to produce generators for properties that are difficult to write explicit generators for. 
% and is also an easy way to write generators that satisfy any decidable precondition.
However, this approach trades ease in development for runtime during generation.
If a property $t$ is satisfied by the generator with probability $p$, then the expected samples (time) to find such an instance is $1/p$---and this can easily be exponential if $t$ occurs rarely enough.
The polynomial time constraint implies that for realistic generators, the use of such rejection sampling is only possible in controlled cases where $1/p$ is at most a polynomial.

We will first see that the polynomial time generable languages are wholly contained within $\NP$.
% Note that this observation was first made in~\citep{sanchis1990complexity}.
\begin{corollary}\label{prop:pg-np}
$\PG{} \subseteq \NP{}$.
\end{corollary}
\begin{proof}
Consequence of \Cref{prop:subset-nondet}.
\end{proof}

\takeaway{\Cref{prop:pg-np} shows that \NP{} membership of the precondition is a necessary condition for efficient generation: if a precondition is \PSPACE{}-hard or \textit{coNP}-hard (e.g.\ quantified formulas, unsatisfiable formulas), there is no hope in seeking a fast, fully general generator---one must resort to rejection sampling or restrictions.}

\begin{example}
If $\NP{} \neq \PSPACE{}$, then the following $\PSPACE$-complete language is not generable in polynomial time:
$$ \textit{QBF} = \{ \phi : \phi\ \text{is a satisfiable quantified boolean formula} \} $$
where a quantified boolean formula is a propositional boolean formula with quantifiers.
\end{example}

\begin{example}
If $\NP{} \neq \textit{coNP}$, then the following $\mathit{coNP}$-complete language is not generable in polynomial time:
$$ \textit{TAUT} = \{ \phi : \phi\ \text{is true for all assignments} \} $$
\end{example}

The natural follow up question is whether the containment $\PG{} \subseteq \NP{}$ is strict.
We see that $\PG{}$ does contain some of the hardest problems in $\NP{}$, namely \textit{CNF-SAT}.

\begin{proposition}\label{prop:cnf-sat-pg}
$\text{CNF-SAT} \in \PG{}$.
\end{proposition}
\begin{proof}
We give a sketch of the key ideas, and defer full detail to \Cref{proof:cnf-sat-pg}.
Consider a generator that first samples a random assignment $\mathbf{x}$, then generates the formula $\phi$ based on $\mathbf{x}$ clause-by-clause.
For each clause, we pick a random literal to evaluate to \textit{True}, referring to the earlier chosen assignment $\mathbf{x}$.
The other literals in the clause are chosen at random.
By construction, as $\phi$ is in CNF, each clause of $\phi$ is satisfied by $\mathbf{x}$, and thus $\phi[\mathbf{x}]$ also evaluates to \textit{True}.

Any satisfiable \textit{CNF-SAT} formula $\psi$ has a chance of being sampled this way: if $\psi$ is satisfied by $\mathbf{y}$, then the above procedure has a chance of choosing $\mathbf{y}$ and then constructing $\psi$ around $\mathbf{y}$.
\end{proof}

\takeaway{A precondition that is hard to \emph{decide} can still be easy to \emph{generate}: you can fully test properties on a \#SAT solver (\Cref{example:counting-sat}) without ever solving SAT, for example by using the generator in \Cref{prop:cnf-sat-pg}.
Thus, testing difficulty cannot be inferred from decision difficulty.}

However, having an efficient generator for an $\NP{}$-hard language does not directly yield a way to sample from easier languages, due to the lack of natural reductions for generators (as far as we know). 
The reduction used in $\NP{}$-hardness embeds the easier language into the harder language.
For generators, we would like the opposite: to surjectively map from instances of the hard language to the easier one, which is what is needed to implement a generator for the easier language.

In fact, despite having a generator for an $\NP{}$-complete problem, we will show below an example of a language in $\NP{}$ (in fact, $P$) that is unlikely to have a polynomial-time generator.

\subsection{\PG{} vs. \NP{}}\label{subsec:pg-np}
We now introduce a concrete problem in \NP{} that is likely not in \PG{}, as its membership in \PG{} would contradict cryptographic assumptions that have withstood decades of scrutiny.

Following \cite{russell1995necessary}, we will base our construction on so-called collision-resistant hash functions.
\begin{definition}[Hash Functions \citep{russell1995necessary}]\label{def:hash-functions}
A hash function $h$ is a collection of functions $(h_n)_{n \in \mathbb{N}}$, $h_n : \{0, 1\}^{n+1} \rightarrow \{0, 1\}^{n}$, such that there exists a computable function $h : \mathbb{N} \times \{0, 1\}^+ \rightarrow \{0, 1\}^+$ with $h(n, x) = h_n(x)$ for all $n \in \mathbb{N}$, and $h$ runs in time polynomial in $n$.
\end{definition}

\begin{definition}[Collision-Resistant Hash Functions]\label{def:collision-resistance}
Hash function $h$ is considered \emph{collision-resistant} if there exists no polynomial-time bounded randomised algorithm $A : \mathbb{N} \rightarrow \{0, 1\}^* \times \{0, 1\}^*$ and polynomial $q$ such that for infinitely many $n \in \mathbb{N}$:

$$\mathbb{P}(A(n) = (x, y) \land x \neq y \land h_n(x) = h_n(y)) > \frac{1}{q(n)}$$
\end{definition}

If a randomised collision-finding algorithm of \Cref{def:collision-resistance} exists for the polynomial $p$-bounded hash function $h$, then one can find size-$n$ hash collisions of $h$ in expected polynomial time. 

\begin{remark}[Unkeyed hashes and uniform attacks]\label{remark:unkeyed-uniform}
In cryptographic literature~\citep{katz2007introduction}, \emph{keyed} hash families $(h_k)_{k \in K}$ are more commonly used, with security against non-uniform attacks.
We instead use \emph{unkeyed} hash functions (\Cref{def:hash-functions}) with security against \emph{uniform} attacks (\Cref{def:collision-resistance}).
Unkeyed hashes capture practical algorithms like SHA-256, but cannot resist non-uniform attacks: a hash collision is guaranteed to exist for each $n$ by the pigeonhole principle, and it is possible to hard-wire them into non-uniform circuits.
Uniform attacks~\citep{rogaway2006formalizing} are a formalisation of the security guarantees in practice: collisions exist, but there exist no efficient algorithms to find them.
As we shall see later, the uniformity of generators and the resistance of unkeyed hashes against uniform attacks are precisely what give us the non-generability result of \Cref{prop:collisions-not-generable}.
\end{remark}

Define the following language for a hash function $h$:
\begin{definition}
For hash function $h$:
$$
\textit{COLLISION}_h = \{(x, y) \in \{0, 1\}^+ \times \{0, 1\}^+:
    x \neq y \land |x| = |y| \land h_{|x|}(x) = h_{|x|}(y)\}
$$
\end{definition}

Clearly, for any hash function $h$ (collision-resistant or not), given $(x, y)$, a Turing machine $M$ can execute $h(x)$ and $h(y)$, and accept if they compute the same hash. 
Since $h$ executes in time polynomial in its input size, $M$ will terminate in polynomial time.
\begin{lemma}
For all hash functions $h$, $\textit{COLLISION}_h \in P$.
\end{lemma}

\begin{proposition}\label{prop:collisions-not-generable}
If for all hash functions $h$, $\textit{COLLISION}_h \in \PG$, then collision resistant hash functions do not exist.
\end{proposition}
\begin{proof}
Suppose $\textit{COLLISION}_h \in \PG$, then there exists a polynomial-time generator $G = (G_i)_{i \in \mathbb{N}}$ with $\Lang{G} = \textit{COLLISION}_h$.
For a given $n \in \mathbb{N}$, we can find a collision of $h_n$ by running $G_{2n+1}$ (assuming 1 character is used as a separator) to obtain a pair $(x, y) \in \textit{COLLISION}_h$ of size $2n+1$. 
Note this is possible due to the uniformity of $(G_i)_{i \in \mathbb{N}}$.
Since $|x| = |y|$ we know that $|x| = |y| = n$, and hence the pair found is a collision of $h_n$: $h_n(x) = h_n(y)$.
As $G_i$ runs in polynomial time, we have therefore a randomised algorithm that finds hash collisions in polynomial time, for infinitely many $n$, and with success probability 1.
\end{proof}

\begin{theorem}[Cryptographic Separation]\label{prop:crhf-implies-pg-neq-np}
If collision-resistant hash functions exist, then $\PG{} \subsetneq \NP{}$. 
In fact, $\Poly{} \not \subseteq \PG{}$.
\end{theorem}

\takeaway{Simple software may be infeasible to test efficiently, in full generality. 
Easy-to-decide preconditions (here, checking a hash collision in \Cref{prop:crhf-implies-pg-neq-np}) can still be difficult to generate satisfying inputs for.
Easy decision complexity is neither necessary nor sufficient for generability.
(For a classification of what is efficiently generable, see our certificate scheme of \Cref{prop:generator-verifier-equivalence}.)}

So, the power of efficient generators has an interesting relationship with the classical decision classes: \PG{} contains some of the hardest problems in \NP{}, yet it also does not contain some problems in \Poly{}.
We thus see that the complexity of the language for generation is a related but different concept from the complexity of the language for decision algorithms.

% \subsubsection{\PG{} With Oracle Access}
Some generators in testing literature also make use of $\textit{SAT}$ solvers~\citep{steinhofel2022input, livinskii2023fuzzing, pschanely2022crosshair, bruni2011peer}, to increase expressivity.
Here we quote a result first shown in~\citep{sanchis1990complexity}: with a $\textit{SAT}$ oracle, one can generate all \NP{} problems in a polynomial number of steps.

Let $\Sigma^P_i$ denote the $i$th level of the polynomial hierarchy with an \NP{} base machine.
For a generator class $X$ with constraints $Y$ and language $L$, define $X^L$ as the class of languages generable by generators with constraints $Y$ and oracle access to a decider for $L$. 
For a generator complexity class $X$ and a decision complexity class $C$, define $X^C = \bigcup_{L \in C} X^L$.

\begin{proposition}\label{prop:pg-oracle-superset-np}
$\Sigma^\Poly_1 = \NP \subseteq \PG{}^\NP{} \subseteq \NP{}^\NP{} = \Sigma^\Poly_2$
\end{proposition}
\begin{proof}
A case from a more general result in~\cite{sanchis1990efficient}.
See \Cref{proof:pg-oracle-superset-np} for a proof of this case.
\end{proof}

\takeaway{Although using a SAT/SMT solver during generation allows \NP{} preconditions to be satisfied, \Cref{prop:pg-oracle-superset-np} shows it is not a panacea: for example, \PSPACE{}-hard or \EXPTIME{}-hard preconditions remain out of reach even with a perfect SAT or SMT solver, respectively.}

\subsection{Expected Polynomial-Time Generators}\label{subsec:epg}
Since sampling is inherently a random process, another feasible resource constraint on generators is via the \emph{expected time} to produce a sample, rather than a deterministic time bound.
The following defines \emph{expected} polynomial time generators:

\begin{definition}\label{def:expected-time-generator-family}
The generator family $G = (G_i)_{i \in \mathbb{N}}$ runs in expected polynomial time if there exists a polynomial $p$ such that for all $i \in \mathbb{N}$, $G_i$ terminates in expected time $p(i)$.
Moreover, for all $w \in \Lang{G_i}$, there exists at least one bitstream where $w$ is produced before time $p(i)$.
\end{definition}
The condition that there exists \emph{at least one} path of polynomial length for any $w$ ensures that there are no fundamentally difficult instances in the language.%
\footnote{Such an issue of hard instances hiding in long runs was also noted in~\cite{jerrum1986random}, where interestingly, the modelling choice made by \citet{jerrum1986random} was to \emph{ignore} the expected-time case and only focus on deterministic time bounds.}
Our definition does not collapse to \PG{}, since execution paths longer than polynomial are still permitted.
Intuitively, expected polynomial time describes repeated retries of a polynomial process with bounded chance of failure: a generator version of \BPP{}.
On lucky runs, the process takes polynomial time, which corresponds to the existence of a polynomial-length path. 

\begin{definition}
The class of languages generable by some expected polynomial-time generator is:
$$
\EPG{} = \{ \Lang{G} : G\text{ is an expected polynomial-time generator} \}
$$
\end{definition}

\begin{proposition}\label{prop:epg-subset-np}
$\PG{} \subseteq \EPG{} \subseteq \NP{}$.
\end{proposition}
\begin{proof}
A generator that terminates in deterministic time $p(i)$ will also terminate in expected time at most $p(i)$, thus $\PG{} \subseteq \EPG{}$.
$\EPG{} \subseteq \NP{}$ follows by simulation---for a full proof, see \Cref{proof:epg-subset-np}.
\end{proof}

Consider the language $\textit{PRIMES}$, defined as follows:
\begin{definition} 
Assume $x$ is interpreted in its binary representation, define:
$$\textit{PRIMES} = \{ x \in \{0, 1\}^* : x > 1 \land \forall y. 1 < y < x \implies gcd(x, y) = 1 \}$$
\end{definition}

We show that $\textit{PRIMES} \in \EPG{}$: an interesting case as it is an open problem in mathematics to determine whether there exists a \PG{} algorithm for $\textit{PRIMES}$~\cite{tao2012deterministic}.
That is, an algorithm which \emph{always} produces a prime number of a certain size within polynomial time.
However, it does have a natural algorithm running in \emph{expected} polynomial time, and also does not follow the structure established for \PG{} in \Cref{prop:generator-verifier-equivalence}.
It would thus be tempting to conjecture that $\PG \subsetneq \EPG$.

\begin{proposition}\label{prop:primes-epg}
$\textit{PRIMES} \in \EPG{}$.
\end{proposition}
\begin{proof}
This is a straightforward application of the fact that $\textit{PRIMES} \in P$~\citep{agrawal2004primes}, and the Prime Number Theorem on the density of primes.
Together, these allow us to use rejection sampling for primes between $2^n$ and $2^{n+1}$.
See \Cref{proof:primes-epg} for a full proof.
\end{proof}

\takeaway{
Giving up worst-case polynomial time guarantees can at times result in simpler generator designs, and potentially allow more sets to be sampled (but is an open problem whether this is the case~\cite{tao2012deterministic}).
However, care must be taken to ensure the runtime is truly polynomial in expectation (as in \Cref{prop:primes-epg}), and that no element is intrinsically hard to sample: otherwise the efficiency and soundness (distributional support) guarantees no longer hold.
}

Akin to boosting the probability of deciding a language in probabilistic complexity classes, we can make a similar statement on boosting the probability of generating an instance for \EPG{}.
\begin{proposition}\label{prop:boosting-epg}
For $L \in \EPG{}$, and any $\epsilon > 0$, there exists an algorithm $A$ that produces a sample $x \in L$ of length $n$ with probability $> 1 - \epsilon$, in polynomial time of $n$ and $\log(1/\epsilon)$.
\end{proposition}
\begin{proof}
By considering long runs and using Markov's inequality to bound the non-generation probability.
See \Cref{proof:boosting-epg} for a full proof.
\end{proof}

This boosting result means that for $\EPG{}$ languages, we can obtain samples with arbitrarily high probability in polynomial time.
Thus, this would also rule out the membership of $\textit{COLLISION}_h$ in $\EPG{}$, unless collision-resistant hash functions do not exist.
\begin{proposition}
If collision-resistant hash functions exist, then $\EPG{} \subsetneq \NP{}$.
Moreover, $P \not \subseteq \EPG{}$.
\end{proposition}
\begin{proof}
Suppose for a contradiction that for all hash functions $h$, $\textit{COLLISION}_h \in \EPG{}$, then there exists an expected polynomial time generator that generates $\textit{COLLISION}_h$.
By \Cref{prop:boosting-epg}, for any $k \in \mathbb{N}$, we can obtain a generator that produces a length-$n$ sample of $\textit{COLLISION}_h$ with probability at least $1/k$, within polynomial time in $n$ (with $\epsilon$ fixed to $1/k$, the $\log(1/\epsilon)$ term becomes a constant factor in the polynomial).
Since this is a procedure that violates the collision resistance property (\Cref{def:collision-resistance}) of $h$, no hash functions are collision resistant, which is a contradiction.
\end{proof}

\section{Space Bounded Generators}\label{sec:space-generators}
Space-bounded generators have a limited amount of space to work with, but can exhibit arbitrary running time.
Whilst this seems like a rather loose restriction on generators, we note that these constraints actually correspond to a large class of generators implemented in practice.

Space-bounded generation  can model test-case generation algorithms that produce not only independent outputs in a single run (like in the case of \PG{} or \EPG{}), but also past-dependent outputs in an iterative fashion.
That is, algorithms where previous outputs can \emph{influence} future produced outputs, for instance by caching some data collected from past outputs to introduce dependence in the generation algorithm.

For a given randomised algorithm $A$ that produces the correlated sequence $(x_1, \ldots, x_n, \ldots)$ one by one, we can define a $G$ that first samples for an $n \in \mathbb{N}$ (e.g., using \Cref{def:naturalencoding}), then simulates $A$ until the first $n-1$ outputs have been emitted, before writing the final $x_n$ on the output tape.
Clearly, $x$ is a possible output of the algorithm $A$ if and only if $x \in \Lang{G}$, and that $G$ uses the same amount of \emph{space} as $A$ to produce $x$.

Many prominent examples of practical generators naturally fit this framework of sequence-producing generators.
These include coverage-guided fuzzing~\citep{bohme2016coverage} and concolic execution \cite{cadar2021klee, godefroid2005dart}---which we will discuss in more detail later on in \Cref{subsec:feedback-generators}.
In these cases, the time taken by the generator is of orthogonal relevance, as many of these methods are designed to run \emph{indefinitely}---for instance, as a background process to continuously find bugs~\citep{bohme2016coverage, busse2020running}.
Thus, for these generation algorithms, space bounds can be a more natural theoretical constraint to study.

There is a caveat with using space-bounded generators verbatim to model these algorithms: many such algorithms call the SUT, which may be arbitrary computation.
Thus, it would make sense to \emph{separate} the generation algorithm from the SUT call: the generation process should be independent of the SUT it is applied to.
This issue will be addressed later in \Cref{subsec:feedback-generators}.

We start by defining the specific space-bounded generation classes considered in this section.
The first is on logspace generators, which have access to pointers and counters up to the size of the generated output.
Logspace computations still make sense, since generators have read-once-only input and write-once-only output, preventing them from using the input/output tapes for storing more information than their space bound permits.
Though seemingly restrictive, we will see later that logspace still suffices for generating instances of useful predicates, like the set of satisfiable 2-SAT formulae or directed graphs with guaranteed connectivity between two chosen nodes.

\iffalse 
\begin{definition}
The generator family $G = (G_i)_{i \in \mathbb{N}}$ runs in logspace if there exists a function $f \in O(\logr(n))$ such that for all $i \in \mathbb{N}$, $G_i$ uses at most $f(i)$ cells on its work tape.
\end{definition}
\fi

\begin{definition}
The class of log-space generable languages are denoted as $\LG{}$.
$$
\LG{} = \bigcup_{c \in \mathbb{N}} \GSPACE(c \logr(n))
$$
\end{definition}

We also consider languages generable within polynomial space.
This class is permissive and covers the majority of the state-dependent generation algorithms seen in practice: storing more than a polynomial amount of data in the size of the test input is impractical in most applications.
\begin{definition}
The class of languages generable within polynomial space is defined as:
$$\PSPACEG{} = \bigcup_{i \in \mathbb{N}} \GSPACE(n^i)$$
\end{definition}

We begin with some general facts on space-bounded generators, connecting them to existing models of time-bounded generators and decision procedures.

The first proposition relates space-bounded decision procedures and time-bounded generators, where a time-bounded generator with \emph{much more} time can generate samples from a space-bounded NDTM.
This is in analogy to a standard result on space-time tradeoffs between NDTMs and DTMs for decision procedures.
\begin{proposition}\label{prop:nspace-subseteq-exptime-generator}
For space-constructible $s : \mathbb{N} \rightarrow \mathbb{N}$,
$$\NSPACE(s(n)) \subseteq \GTIME(2^{O(s(n))})$$
\end{proposition}
\begin{proof}
Deferred to \Cref{proof:nspace-subseteq-exptime-generator}.
Intuitively, the extra time on the generator allows for the nondeterministic machine to be simulated by enumerating configurations.
Then, by using a complete decision language within $\NSPACE(s(n))$ (itself decidable on the generator as a procedure) that decides whether an accepting suffix exists for $M$, it is possible to generate, for any $L \in \NSPACE(s(n))$, any $x \in L$ via repeated queries to this complete language.
\end{proof}

\Cref{prop:nspace-subseteq-exptime-generator} together with \Cref{prop:subset-nondet} relates space and (much higher) time bounds for generators, analogous to the classic time-space tradeoff theorem for deciders.

\begin{corollary}[Time-Space Tradeoff]\label{prop:generator-time-space-tradeoff}
For a space constructible function $s : \mathbb{N} \rightarrow \mathbb{N}$ with $s(n) \geq \logr(n)$, 
$$\GSPACE(s(n)) \subseteq \GTIME(2^{O(s(n))})$$
\end{corollary}

For the decision version of \Cref{prop:generator-time-space-tradeoff}, a configuration-counting argument is the standard proof. 
This argument does not carry over immediately, since space-bounded generators can be driven into cycles for an unbounded amount of time by their \emph{input}.
In \Cref{appendix:config-counting-gspace-gtime}, we discuss the issue with the configuration counting argument in more detail and give a fix in an alternative, direct proof of \Cref{prop:generator-time-space-tradeoff}.

\subsection{Logspace Generators}\label{subsec:logspace-generators}
\Cref{prop:subset-nondet} and \Cref{prop:nspace-subseteq-exptime-generator} together mean that every \NL{} language is generable in polynomial time (i.e., has a \PG{} algorithm).

\begin{corollary}
$\LG \subseteq \NL \subseteq \PG$.
\end{corollary}

This implies properties with preconditions such as 2-SAT, or connected directed graphs, will also have efficient polynomial-time generators.
This also gives a practical sufficient condition for \PG{} membership: if the predicate can be decided by an \NL{} machine, then an efficient polynomial-time generation algorithm is guaranteed to exist. 

We can apply a known result~\citep{arenas2019efficient} to show that there exists an algorithm in \EPG{} uniformly generating any $L \in \LG{}$.
The main idea is to view logspace generators as logspace transducers (adapted to our setting), which \citet{arenas2019efficient} proved have uniform samplers with a bounded probability of failure.
It is straightforward to adapt such uniform samplers to a generator in \EPG{}.

\begin{proposition}\label{prop:lg-epg}
If $L \in \LG{}$, then there exists a \EPG{} generator of uniform distributions. % for each $L_i, i \in \mathbb{N}$.
\end{proposition}
\begin{proof}
We defer the proof to \Cref{proof:lg-epg}.
\end{proof}

\takeaway{When a precondition is decidable in \NL{} (e.g.\ reachability, 2-SAT, and most syntactic well-formedness checks) an efficient generator is \emph{guaranteed} to exist, and (by \Cref{prop:lg-epg}) even a uniform one.}

We also see that even in the restricted setting of \LG{}, we are still able to generate \NL{}-complete languages, analogous to how \PG{} contains \NP{}-complete problems.

\begin{definition}
Reachability on directed graphs is defined as:
$$\REACH{} = \{ \langle G, s, t \rangle : \exists\ \text{path from}\ s\ \text{to}\ t\ \text{in}\ G\}$$
\end{definition}

\begin{proposition}\label{prop:reach-lg}
$\REACH{} \in \LG{}$.
\end{proposition}
\begin{proof}
The construction is based on certificate sampling and instance reconstruction, which is generalised by \Cref{prop:generator-verifier-equivalence} later on.
We give a sketch here and defer the details to \Cref{proof:reach-lg}.
We represent graphs as an edge-list, and partition \REACH{} by the total number of vertices and edges.
First, select the source and target as $s$ and $t$, then generate a random walk from $s$ to $t$, outputting edges as they are used.
Then, add edges to the graph in Erdos-Renyi fashion, again outputting them directly after selection.
Every $x \in $ \REACH{} has a nonzero probability of being sampled.
\end{proof}

\subsection{Generators With Feedback}\label{subsec:feedback-generators}
Testing techniques like concolic execution and coverage-guided fuzzing produce a (potentially infinitely long) sequence of test inputs, where the results from executing previous inputs are used to guide the generation of the next using the results from running the SUT.
Since the SUT can perform arbitrary computation, we separate the SUT from the \emph{generation technique} in our model.

\begin{definition}[Function Oracle Turing Machine]
A TM $M_f$ with oracle access to $f : \Sigma^* \rightarrow \Sigma^*$ is a standard TM with an additional read-write query tape $t_q$, a read-only answer tape $t_a$, and a designated state $s_o$.
When $M_f$ enters state $s_o$ with $x$ written on the query tape $t_q$, the contents on the answer tape $t_a$ will be replaced by $f(x)$. 
The language recognised/generated by $M_f$ is denoted $\Lang{M_f}$.

For a Turing machine $M$, we denote $M_f$ as the same Turing machine with access to an oracle for $f$.
When there is no ambiguity, we may drop the subscript $f$ from $M_f$ for brevity.
If $M$ is a generator, we will denote it with $G$ instead, i.e. $G_f$ is a generator $G$ with oracle access to $f$.
\end{definition}
This oracle machine is different from the ones discussed in \Cref{prop:pg-oracle-superset-np}: here oracle access is to a function $f$ whereas the previous one gives oracle access to a decider $f : \Sigma^* \rightarrow \{T, F\}$.

The execution of the SUT, together with the computation of data guiding future generation, can be viewed as an oracle available to the generator.
We present simplified versions of two prevalent testing methodologies and show how they fit into our framework, noting that real implementations often employ variants and optimisations to reduce time and space usage.
Nonetheless, our eventual result still holds: any space-bounded generation process, no matter how it is guided, can only generate languages within some complexity class, thus ruling out all languages outside of the class.

\begin{example}[Concolic Execution]\label{example:concolic}
Concolic execution for an SUT $S$ is an iterative process that first generates an input $x_0$, then tracks the control flow path taken by the run $S(x_0)$, collecting path constraints $p_0$ along the way.
Then, new inputs $x_{i}$ are generated to ensure that they reach control-flow paths unexplored by previous inputs $x_0, \ldots, x_{i-1}$, usually through an SMT solver to produce $x_i$ that satisfies an alternative path constraint, which forces unexplored paths to be taken.

\ifarxiv
Here, the SUT oracle \(f : \Sigma^* \to \Sigma^*\) accepts a generated input \(x_i\), executes the SUT on \(x_i\), and returns the path information \(p_i = f(x_i)\) observed during that execution. 
The generator $G$ has oracle access to $f$.
After generating each output $x_i$, $G$ calls $f(x_i)$ and stores the paths taken within the codebase as a list $p_0 = f(x_0), \ldots, p_{i} = f(x_{i})$.
To produce the element $x_{i+1}$, $G$ invokes a solver to find an input $x_{i+1}$ that satisfies the path condition $p_{i+1}$, distinct from all previous paths $p_0, \ldots, p_{i}$.
\else
In our framework, the SUT oracle $f : \Sigma^* \to \Sigma^*$ maps a generated input $x_i$ to the path information $p_i = f(x_i)$ observed during its execution.
The generator stores $p_0, \ldots, p_i$ and invokes a solver to find an $x_{i+1}$ that satisfies a path condition distinct from all of them.
\fi
\end{example}

\begin{example}[coverage-guided fuzzing]\label{example:cgf}
coverage-guided fuzzing produces new inputs by mutating previous inputs that are likely to gain new coverage.
Historical input $x$ is selectively stored as $y_k$ together with the code coverage $p_k$ obtained by $y_k$.
To generate $x_{i+1}$ that may achieve new coverage, inputs $y$ that were close to achieving new coverage are selected for mutation.
If new coverage has been found by $x_{i+1}$, then it is stored along with the previous inputs that expanded coverage: $y_1, \ldots, y_k$.

\ifarxiv
Here, the SUT oracle $f : \Sigma^* \rightarrow \Sigma^*$ accepts inputs $x_i$ and outputs the coverage data $p_k$ for the codebase achieved by $x_i$.
The generator has oracle access to $f$, storing historical inputs $y_1, \ldots, y_k$.
Using these, it produces new inputs by mutating existing inputs for new coverage.
\else
In our framework, the oracle $f : \Sigma^* \rightarrow \Sigma^*$ maps an input to its coverage data.
Then, the generator stores the coverage-expanding inputs $y_1, \ldots, y_k$, and mutates them to produce new inputs.
\fi
\end{example}

Since the oracle contains computation of the SUT, they can realistically be of any complexity (e.g. compilers), and is independent of the generation complexity.
Therefore, we do not constrain resource usage of the oracle.
However, as the output still needs to be parsed and stored by the generator, the \emph{size} of the output from the oracle call can affect the complexity of generation.
For an oracle, we therefore define the following constraint:
\begin{definition}
A function $f : \Sigma^* \rightarrow \Sigma^*$ is length-bounded by $s : \mathbb{N} \rightarrow \mathbb{N}$ if for all $x \in \Sigma^*$, $|f(x)| \leq s(|x|)$.
\end{definition}

Complexity classes can be defined analogously for Turing machines with oracle access.
Here, the oracle input tape is under the same space bounds as work tapes.
\begin{definition}
For space constructible function $s : \mathbb{N} \rightarrow \mathbb{N}$, and a function $f : \Sigma^* \rightarrow \Sigma^*$, define the classes of languages decidable/generable within $s$ space as follows:
\begin{align*}
    \DSPACE{}(s(n))^f &= \{ \Lang{M_f} : M \text{ is space-bounded by $s$ } \} \\
    \GSPACE{}(s(n))^f &= \{ \Lang{G_f} : G \text{ is space-bounded by $s$ } \}
\end{align*}
Note that the space-bound $s$ applies both to the work tape and the query tape.
\end{definition}

\begin{definition}
For a function $f : \Sigma^* \rightarrow \Sigma^*$,
$$
    \PSPACE{}^f = \bigcup_{c \in \mathbb{N}} { \DSPACE(n^c) }^f \qquad
    \PSPACEG{}^f = \bigcup_{c \in \mathbb{N}} { \GSPACE(n^c) }^f
$$
\end{definition}

Both techniques described in \Cref{example:concolic} and \Cref{example:cgf} can use infinite space in the limit.
However, real-world applications are bounded by resource limits, and using polynomial space as a proxy for realism, we ask: what constraints placed on coverage-guided fuzzing and concolic execution ensure that the procedures remain in $\PSPACEG{}$? 

\paragraph{Coverage-Guided Fuzzing}
Since the codebase is fixed, the coverage information has constant size regardless of the input sizes, thus the output of the oracle is always within a polynomial bound (in fact, constant size).
Each new coverage-finding input is stored alongside its coverage information, which consumes space on the work tape.
Within a polynomial space bound $p$, generators of outputs of sizes up to $n$ can store $p(n) / n$ previous inputs (if stored naively). 

\paragraph{Concolic Execution}
The data returned by the oracle consist of the path taken by the generated input $x_i$ of size (up to) $n$.
For the oracle to stay within polynomial space bounds, the path taken must therefore have at most polynomially many control flow statements (which is achievable with instrumentation and timeout, for instance).
As the previous paths taken are stored by the generator, there can be at most a polynomial number of historical paths stored.

The above examples are illustrative for demonstrating the kind of constraints that polynomial space induces on naive implementations of popular generation algorithms.
In reality, optimisations would be used on implementations for either technique~\cite{rebert2014optimizing, avgerinos2014enhancing}: the calculus shifts accordingly for the number of paths or test cases one is able to store, but the bottleneck remains in the polynomial space bounds.
We show below the limits of generation capability for such feedback-driven approaches when a polynomial space bound is adhered to. 

The previous complexity results for space-bounded generators still hold with oracle access.
\begin{proposition}\label{prop:gspace-subset-nspace-subset-exptime-oracle}
Suppose $s : \mathbb{N} \rightarrow \mathbb{N}$ is both space and time constructible with $s(n) \geq \logr(n)$, and $f : \Sigma^* \rightarrow \Sigma^*$ is length-bounded by $2^{O(n)}$.
Then: 
$$\GSPACE(s(n))^f \subseteq \NSPACE(s(n))^f \subseteq \GTIME(2^{O(s(n))})^f$$
\end{proposition}
\begin{proof}
Deferred to \Cref{proof:gspace-subset-nspace-subset-exptime-oracle}.
Similar to the non-oracle version of \Cref{prop:nspace-subseteq-exptime-generator}.
Extra care is needed when counting the configurations of oracle machines: the read-only answer tape and $f$ being a function are the two factors that keep the number of configurations at $2^{O(s(n))}$.
\end{proof}

\begin{proposition}\label{prop:pspaceg-eq-pspace-oracle}
For $f : \Sigma^* \rightarrow \Sigma^*$ length-bounded by $2^{O(n)}$,
$\PSPACEG{}^f = \PSPACE{}^f$
\end{proposition}
\begin{proof}
Deferred to \Cref{proof:pspaceg-eq-pspace-oracle}.
$\PSPACEG^f \subseteq \PSPACE^f$ follows from a variant of Savitch's theorem adapted for our oracle machines.
$\PSPACE^f \subseteq \PSPACEG^f$ follows since rejection sampling can be performed whilst reusing space.
\end{proof}

If the oracle $f$ itself is computable in polynomial space, then both $\PSPACE^f$ and $\PSPACEG^f$ collapse to $\PSPACE$.
We thus have the following as a corollary:
\begin{theorem}\label{prop:pspaceg-eq-pspace}
$\PSPACEG{} = \PSPACE{}$.
\end{theorem}

If the SUT itself runs in polynomial space, then the polynomial-space generable properties for the SUT are exactly the $\PSPACE$ decidable ones. 
If the SUT does not run in polynomial space, then the polynomial-space generable properties are exactly those with input sets decidable by a $\PSPACE$ machine with oracle access to the SUT.

\takeaway{\Cref{prop:pspaceg-eq-pspace} implies feedback-driven techniques---coverage-guided fuzzing, concolic execution---are bounded by \PSPACE{}. 
For example, unless the SUT itself is powerful enough to solve \EXPTIME{}-complete problems, \EXPTIME{}-complete sets remain out of reach regardless of how the feedback loop is engineered.
\Cref{prop:generator-time-space-tradeoff} says that any input set a feedback-driven technique can sample in polynomial space is also reachable by a (slower) generator of independent samples.}

\begin{example}
If $\PSPACE{} \neq \textit{EXP}$, then the following problems are not generable in polynomial space.
\begin{enumerate}
    \item Board positions of generalised size-$n$ Chess/Shogi/Draughts games with a winning strategy for the player that moves first.
    \item For any programming language, the set of programs that terminate in $k$ steps, with the step count specified in binary.
\end{enumerate}
\end{example}

\section{Designing Generators and Generator Libraries}\label{sec:practitioners}
In this section, we introduce some practical insights that come from our theoretical framework.
We first provide canonical constructions for building efficient generators (that are members of \PG{}), and then apply the results we established in generation complexity to gain understanding of the limitations of compositional library designs for randomised testing.

\subsection{Conditions for \PG{} Membership}\label{subsec:pg-conditions}
We propose two characterisations of \PG{} languages, which formalise the main known techniques for obtaining \PG{} generators.

\subsubsection{$\PG{}$ Via Basis Extension}
The first technique is based on the ability to generate \emph{default} elements of the language.
A difficulty with generating inputs accepted by \NP{} deciders is that the majority of the branches may fail.
Thus, a generator would need to find a branch of the \NP{} computation that succeeds in order to confirm a string's membership---which by naive methods leaves an exponential number of branches needing to be checked.
However, if one has the ability to efficiently find a default element as a fallback, a generator is free to try potentially-failing computations, allowing it to generate the remainder of the \NP{} language.

\begin{definition}
For the language $L$, define the functions that compute a default element of $L$ as:
$$
\DEFAULT(L) = \{ f : \{1\}^* \rightarrow L : L_i \neq \emptyset \implies f(1^i) \in L_i, L_i = \emptyset \implies f(1^i) = \epsilon \}
$$
\end{definition}

\begin{proposition}\label{prop:default-elem-implies-pg-membership}

For $L \in \NP$, $\DEFAULT(L) \cap \mathit{FP} \neq \emptyset$ if and only if $L \in \PG{}$, where $\mathit{FP}$ denotes polynomial-time computable functions. 
\end{proposition}
\begin{proof}
The main idea is that being able to compute a default element $x$ means we can simulate a branch of the NDTM, falling back to $x$ if the branch rejects. 
For a full proof, see \Cref{proof:default-elem-implies-pg-membership}.
\end{proof}

We can extend the structure of such techniques to a sufficient condition for generating samples from an \NP{} language, by asserting extra structural conditions on the language for efficient extension from a base element. 

\iffalse
\begin{proposition} %???
If $L \in RELP$ is self-reducible, then for any $x$ the lexicographically minimal solution in $y_{\mathit{min}} \in Sol(x)$ is computable in polynomial time.
\end{proposition}

\begin{corollary} %???
If $L \in P$ is self-reducible, then $L \in \PG$.
\end{corollary}
\fi

\begin{proposition}\label{prop:sufficient-pg-basis}
    Let $L$ be a language and $\preccurlyeq$ be a polynomial-time computable partial order on $\Sigma^*$. 
    Let $K = \{s \in L : \nexists s' \in L \setminus \{s\}, s' \preccurlyeq s\}$ be the set of minimal elements in $L$ according to $\preccurlyeq$.
    $L \in \PG{}$ if the following conditions are met:
    \begin{enumerate}
        \item \textbf{Upward Closure}: For any $s \in L$ and string $s'$, if $s \preccurlyeq s'$, then $s' \in L$.
        \item \textbf{Samplable Minimal Basis}: $K$ contains every minimal element of $L$, and is in \PG{}.
        \item \textbf{Efficient Extension}:
         For every element $s \in K$, there exists a polynomial-time generator $G_s$ such that 
         $\Lang{G_s} = \{s' : s \preccurlyeq s' \land |s'| = |s|\}$.
    \end{enumerate}
In other words, if there exists a core set of instances in $L$ that is efficiently generable, and one can efficiently extend from the core instances to any instances in $L$, then the set is efficiently generable.
\end{proposition}
\begin{proof}
Deferred to \Cref{proof:sufficient-pg-basis}.
\end{proof}

This provides a technique for obtaining \PG{} sampling algorithms, and indeed many \NP{}-hard languages satisfy this property. 
We give examples applications of \Cref{prop:sufficient-pg-basis} in \Cref{example:basis-extension-examples}.

\subsubsection{\PG{} Via Certificate Sampling}
Many \NP{} decision problems are hard because of the need to find a potentially small set of certificates for the instances.
However, we have seen that many hard \NP{} problems nonetheless have efficient generators.
The problem of generating satisfying instances of an \NP{} language becomes easier than decision when there is a means of constructing instances from certificates: ensuring that the generated instance is in the language by construction.

This is the core insight used by many existing generators of \NP{}-complete languages, which is based on their verifiers: first sample a random certificate, and then sample a random instance that is verified by the certificate.
We generalise this and show that the reverse also holds.
That is, every generator for $L$ is associated with some canonical verifier $V$ for $L$, a method to sample certificates of $V$, and a method to recover random instances from certificates. 
Thus, every generator can be viewed as an efficient algorithm to recover instances for each canonical certificate (i.e., the certificate consisting of bitstrings).

\begin{definition}[Polynomial-Time Verifier]
$L$ has a polynomial-time verifier $V$ if:
\begin{itemize}
    \item $V$ takes inputs $w\#c$, where $w \in \Sigma^*$ and $c \in \Sigma^+$ is a certificate.
    \item $V$ terminates in polynomial time of its input size. 
    \item There exists a polynomial $p$ where:
    \begin{itemize}
        \item If $w \in L$, then there exists a certificate $c$ with $|c| \leq p(|w|)$ such that $V$ accepts $w\#c$. 
        \item If $w \notin L$, then for all certificates $c$, $V$ rejects $w \# c$.
    \end{itemize}
\end{itemize}
\end{definition}

\begin{definition}[Certified Inputs]
For a verifier $V$, and a certificate $c$, define the certified inputs of $c$, $V_f(c)$ as:
$$V_f(c) = \{ w \in \Sigma^+ : V(w\#c) \text{\ accepts} \}$$
\end{definition}

\begin{definition}[Certificate Generator] 
Verifier $V$ for $L$ has a certificate generator $S_V$ if there exists a polynomial $q$, and $S_V = (S_n)_{n \in \mathbb{N}}$ is a uniform family of generators where:
\begin{enumerate}
    \item Each $S_i$ is such that $L_i = \Sigma^{i} \cap \bigcup_{c \in \Lang{S_i}} V_f(c)$.
    \item Each $S_i$ runs in time $q(i)$.
    \item If $L_i$ is empty, then $S_i$ rejects.
\end{enumerate}
\end{definition}

Note that in the first condition, we constrain $\bigcup_{c \in \Lang{S_i}} V_f(c)$ to the subset of length-$i$ strings.
This is necessary as a certificate may verify instances of many lengths.
As a concrete example, consider a verifier of SAT, taking certificates to be the variable assignment.
A specific assignment can be the certificate for formulas of arbitrary length.

\begin{definition}[Certificate Recovery]
Verifier $V$ for $L$ has a certificate recoverer $R_V$ if there exists a polynomial $p$ in two variables, and $R_V = (R_i)_{i \in \mathbb{N}}$ is a uniform family of TMs where:
\begin{enumerate}
    \item Each $R_i$ has a read-only input tape for certificates $c \in \Sigma^+$, a read-only input tape of random bits, and an output tape.
    \item For all $c$, $R_i(c)$ is a generator for $L_i \cap V_f(c)$.
    \item For all $c$, $R_i(c)$ runs in time $p(i, |c|)$.
\end{enumerate}
\end{definition}

\begin{definition}[Certificate Scheme]
$V$ has a certificate scheme $(S_V, R_V)$ if there exist a certificate generator $S_V$ and also a certificate recoverer $R_V$.
\end{definition}

\begin{theorem}[Generators are Certificate Schemes]\label{prop:generator-verifier-equivalence}
The following are equivalent:
\begin{enumerate}
    \item \label{cert-gen} $L$ is generable in polynomial time ($L \in \PG{}$).
    \item \label{cert-verif} There exists a polynomial-time verifier $V$ for $L$, and $V$ has a certificate scheme. 
\end{enumerate}
\end{theorem}
\begin{proof}
We defer the full proof to \Cref{proof:generator-verifier-equivalence}.
Intuitively, a certificate scheme yields a polynomial time generator by composing the components together.
Conversely, given a $\PG{}$ generator $G$, one can derive a verifier $V_G$ by treating the certificate as the input bitstring of $G$.
Then, one can validate that there is a certificate scheme for $G$ based on the verifier $V_G$, by sampling for a random bitstring certificate of length $p(n)$, where $p$ is the polynomial time-bound of $G$.
\end{proof}

\takeaway{A strategy for building an efficient generator for a language is to study its \emph{verifier} rather than its decider, by sampling a certificate first, then building a random instance around it.
\Cref{prop:generator-verifier-equivalence} shows that \emph{every} efficient generator actually uses some instance of this strategy.}

\subsection{Compositional Libraries For Generators}\label{subsec:compositional-libraries}
As established earlier, while rejection sampling gives an easy way to automatically derive generators that satisfy a specific predicate, it can quickly become intractable as the probabilities of data satisfying the predicate drop. 
Moreover, if the predicate is difficult to verify, then the validation step during test-case generation can itself become the dominant factor for a testing campaign.

Due to this, property-based testing libraries often also provide combinators for users to specify custom generators, rather than relying on rejection sampling to produce inputs.
A desirable characteristic for such libraries is \emph{compositionality}.
In the context of testing, one clear place for compositionality involves the duality of generators and predicates: predicates and generators both correspond to formal languages on the inputs.
Where predicates interpret a language $L$ as a function of type $\Sigma^* \rightarrow \mathbb{B}$, generators interpret $L$ as a surjective computable function $\{0, 1\}^* \rightarrow L$ that constructs instances of the language from randomness.

Logical predicates are easily specified using existing languages, such as propositional logic, described by the following grammar, where $\mathit{Atomic}$ refer to base predicates taken as fact (e.g. $\leq, \equiv$):
$$
    T := \mathit{Atomic}\ |\ T \land T\ |\ T \lor T\ |\ T \rightarrow T\ |\ \neg T 
$$

Compositionality in this context is the requirement for a function $f$ to map the logical language into generators homomorphically, such that $f$ follows the same structure as the logical language:
\begin{gather*}
    f(a : \mathit{Atomic}) = f_a(a) \qquad f(t_1 \rightarrow t_2) = \mathit{gImplies}(f(t_1), f(t_2)) \qquad f(\neg t) = \mathit{gNot}(f(t)) \\
    f(t_1 \land t_2) = \mathit{gAnd}(f(t_1),  f(t_2)) \qquad f(t_1 \lor t_2) = \mathit{gOr}(f(t_1), f(t_2))
\end{gather*}
where $f_a$ maps the base predicate $a$ to a generator of values that satisfies the predicate $a$, and $\mathit{gAnd}, \mathit{gOr}, \mathit{gImplies}, \mathit{gNot}$ are the functions that implement the generator homomorphisms for logical operators $\land, \lor, \rightarrow, \neg$ respectively.

We now show that, based on intuitions imported from complexity theory via our formalisation, such a compositional design for efficient generators is not possible in general---if the logical language for predicates can express predicates equivalent in power to the operations $\neg$ or $\land$.

To put it another way: in a compositional library of generators/predicates, either the generators cannot have polynomial runtime guarantees, or the logical operations must be limited in expressivity---it cannot contain anything equivalent to $\neg$ or $\land$ (that is, $\mathit{gAnd}$ and $\mathit{gNot}$ cannot have implementations for generators).

Firstly, we show that some operators can be implemented composably for efficient generators.
\begin{proposition}\label{prop:pg-closed-ops}
\PG{} is closed under concatenation, Kleene star, and union. 
\end{proposition}
\begin{proof}
We defer the proof to \Cref{proof:pg-closed-ops}.
\end{proof}

On the other hand, we see that some logical operators on predicates cannot have composable implementations on generators (under standard complexity constraints).
\begin{theorem}[Efficient Generators are not Composable]\label{prop:pg-not-closed-intersection}
If $\NP \neq \mathit{coNP}$, then \PG{} is not closed under complementation.
If collision-resistant hash functions exist, then \PG{} is not closed under intersection.
\end{theorem}
\begin{proof}
Complementation follows from \Cref{prop:pg-np} and \Cref{prop:cnf-sat-pg}.
If $\text{CNF-UNSAT} = \text{CNF-SAT}^c \in \PG{}$, then $\text{CNF-UNSAT} \in \NP{}$ by \Cref{prop:pg-np}. 
Thus $\NP = \mathit{coNP}$ and we have a contradiction.
See \Cref{proof:pg-not-closed-intersection} for a proof of the intersection case.
\end{proof}

\takeaway{No PBT library can offer a general \textit{gAnd} or \textit{gNot} combinator with polynomial-time efficiency guarantees. 
\Cref{prop:pg-not-closed-intersection} means library designers must choose: restrict the predicate language (e.g.\ to \NL{}/linear Datalog, where compositional compilation is possible), or accept rejection sampling---which our results show is then close to optimal.}

\begin{remark}[What is still possible?]
\Cref{prop:pg-not-closed-intersection} does not rule out compositionality on restricted predicates, or on alternative logical languages that do not entail intersection and negation operators.
For instance, we know that $\NL \subseteq \PG$, and \NL{} is closed under the logical operations discussed in this section (complementation, intersection, union).
Thus, it would be feasible in theory to design a composable generator library, restricted to \NL{} predicates $\mathit{Pred} : L \rightarrow \mathbb{B}$ that produces polynomial-time composable generators for \NL{} languages.
Descriptive complexity gives a path to achieving this: \NL{} Turing machines have an exact correspondence with linear Datalog programs.
Thus, it would be possible in theory to write a compiler that composes linear Datalog predicates together, and compiles them into polynomial-time generators of their accepted languages.
\end{remark}

\section{Related Work}\label{sec:related-work}

The closest work to ours is \citep{sanchis1990efficient}, which introduced complexity classes for polynomial-time generators and observed several similar results to ours: polynomial-generable languages being in \NP{}, some \NP{}-complete problems are polynomial-time generable, and placing oracle-based generation within the polynomial hierarchy. 
They also show a conditional separation for $\PG{} \subsetneq \NP{}$, assuming a tally language (one with at most one string per length $n$) exists in $\NP{} \setminus \Poly$---or equivalently, $E \neq \mathit{NE}$.
Tally language generators, however, have no realistic interpretation in testing: they are decision algorithms in disguise that cannot produce different random instances. % make no meaningful use of the generator's randomness. 
In \Cref{subsec:pg-np} we show that non-sparse languages can also be non-generable, under a cryptographic assumption not known to be related to the sparse languages assumption.
This both elucidates generability for the testing context and offer independent complexity-theoretic intuition linking generability and cryptography. 
Moreover, \citep{sanchis1990complexity, sanchis1990efficient} do not consider other resource constraints (unconstrained, expected polynomial time, log-space and polynomial-space bounds), generability conditions via core languages and verifiers, or connections to cryptography---unsurprisingly, perhaps, as their work predates PBT, a common application of randomised testing. 
Our work can thus be viewed as both an extension and a modern reinterpretation of their earlier work, with practical applications to randomised testing.

A classic line of work~\citep{jerrum1986random,jerrum2003counting} studies \emph{uniform} generation and repetition-free enumeration in a different setting: given a relation $R \subseteq X \times Y$ and an $x \in X$, one seeks a sampler for $\mathit{Sol}(x) = \{y \in Y : (x, y) \in R\}$. 
\citet{jerrum1986random} showed that for self-reducible $R$, almost uniform generation of $\mathit{Sol}(x)$ is as hard as approximately counting $|\mathit{Sol}(x)|$; more recently, \citet{arenas2019efficient} showed that if $R$ is checkable in logspace, then $\mathit{Sol}(x)$ is almost-uniformly sampleable, approximately countable, and enumerable with polynomial delay. 
Their setting differs from ours in two ways: they sample certificates for instances $x \in L$, whereas we sample instances of $L$ itself; and they require uniformity while permitting failure, whereas we require only full support on $L$ but forbid failure. 
The two lines of work nonetheless intersect in \Cref{prop:lg-epg}, where interpreting logspace generators in the framework of \citet{arenas2019efficient} yields the expected polynomial-time uniform generability for \LG{} languages. 
Our framework is currently distribution agnostic, and extending it to prescribed distributions on $L$, perhaps along the lines of \citet{jerrum1986random}, is an avenue for future work.

Levin's theory of average-case complexity~\citep{levin1986average} studies the hardness of problems whose instances are drawn from \emph{P-samplable} distributions~\citep{impagliazzo1990no}.
A polynomial-time generator $G = (G_n)_{n \in \mathbb{N}}$ for $L$ implements a full-support P-samplable distribution for $L$, by composing the $G_n$ with a distribution supported on $\mathbb{N}$.
The focus differs, however: \citet{levin1986average} studies the hardness of deciding $L$ given \emph{some} distribution $\mu$ (not necessarily full-support), whereas we ask which languages \emph{admit} a full-support sampler at all, and what structure it must have (\Cref{subsec:pg-conditions}).
Average-case hardness of language-distribution pairs also underpins cryptographic security, e.g.\ in constructing one-way functions.
Generating instances around planted witnesses in \Cref{prop:cnf-sat-pg} for SAT (first observed in~\citep{sanchis1990complexity}) is an instance of the classic planted-assignment construction~\cite{feldman2015complexity, juels2000hiding}.
In \Cref{prop:generator-verifier-equivalence}, we show that every polynomial-time generator is a certificate scheme: witness planting is not merely a technique for efficient generation, but in some sense the only one.

The duality between predicates and generators has previously been investigated in a programming languages and verification context.
In~\citep{claessen2014generating, lampropoulos2017generating}, generators are derived from implementations of deciders by techniques based on \emph{narrowing} from logic programming folklore.
In our notation, for predicate $P$, the derived generator $G_P$ generates the same language that $P$ recognises: $\Lang{G_P} = \Lang{P}$.
However, the unconstrained use of rejection sampling in these derivations means the correspondence carries no complexity guarantees, and \Cref{subsec:compositional-libraries} limits the complexity of generators from such constructions: no compositional mapping from a logical language containing conjunction or negation to generators can preserve polynomial-time generation guarantees, under standard assumptions.
Libraries have also been developed to facilitate proofs on generators within theorem provers~\citep{paraskevopoulou2015foundational}.
However, these proofs are on a per-generator basis, primarily on soundness (corresponding to our notion of full support on $L$), and do not make statements on the complexity of generation---which is our focus.

\ifarxiv
% Generators beyond testing
Although software testing was the motivation for our work, random generators of the kind we study are used throughout computational sciences.
Many fields rely on non-trivial generators of random data, often as the core part of a larger algorithm or methodology.
\emph{Markov Chain Monte Carlo} (MCMC) is a widely used technique to produce random samples from sets whose membership has a high decision complexity~\cite{roy2020convergence}.
Designing a Markov chain that converges to the desired distribution on some set can be viewed as a generation problem in our formalisation.
Our theory of space-bounded generation for correlated outputs (\Cref{sec:space-generators}) yields statements on the feasibility of designing chains for a given set.
Chemical space generation is another field where combinatorial databases for molecules are built, with applications to drug discovery and reaction pathways.
For the database to be useful, its elements must satisfy graph-theoretic constraints of molecules~\cite{gugisch2015molgen, ruddigkeit2012enumeration}, feasibility constraints from physical simulations~\cite{yu2021novel, hollingsworth2018molecular}, and synthesisability constraints imposed by the domain of application~\cite{levin2023computer}.
The combinatorial explosion makes throughput critical for these databases, which are often accessed through sampling rather than enumeration. 
Our theory provides results connecting the complexity of the constraints (defining the set of molecules) with both the feasibility and efficiency of building such a chemical space (\Cref{sec:general-statements}, \Cref{sec:polytime-generators} and \Cref{sec:space-generators}).
Many more examples exist.
In each case, the generator is the core component of a larger pipeline producing randomised outputs, drawn from a set with nontrivial membership constraints: exactly the central object of study in this work.
\fi

\section{Conclusion}\label{sec:conclusion}
We formalised generability as a complexity-theoretic notion by modelling generators as resource-bounded transducers, and studied how constraints shape what data can be generated. 
Without resource bounds, we see that generable languages coincide exactly with the recursively enumerable languages, while under polynomial-time bounds we obtain a structured hierarchy of generability classes contained within \NP{}, with logspace being a more restrictive condition than polynomial-time for generators. 
Moreover, we find that space-based constraints can be applied naturally to generators of \emph{correlated} inputs, and derive the limits of their expressivity.
We also give characterisations of the polynomial-time generable languages, showing they correspond exactly to a subclass of verifiers.
Applying our theory to testing libraries, our framework can rule out whole classes of designs from the realm of possibility.
Altogether, our results provide a foundational framework for reasoning about generators in property-based testing using tools from complexity theory. 
More broadly, our approach also opens avenues for transferring complexity-theoretic techniques to the study of testing, while exposing potential new complexity questions centered on generation.

\bibliographystyle{ACM-Reference-Format}
\bibliography{references}

\newpage
\appendix
\renewcommand{\thesection}{A.\arabic{section}}
\setcounter{section}{0}
\crefalias{section}{appendix}
\section{Alternative notion of eventual productivity}\label{appendix:alternative-eventually-productive}
Another candidate for the notion of ``eventually productive'' may be the following---which in any case would be a desirable necessary property for generators:
\begin{definition}[Extensible Transducer]\label{def:extensible-transducer}
    A bitstring transducer $T$ is extensible if for any $b \in \{0, 1\}^*$ where $T(b)$ rejects, there exists an $e \in \{0, 1\}^+$ such that $T(b \concat e)$ accepts.
\end{definition}
On the surface, this seems to describe the \emph{eventually productive} property, since no computations are ever ``dead'', due to the extension $e$ being always available for any rejected $b$.

However, it is possible that the extension $e$ becomes increasingly rare as the computation advances.
Consider the following case:
\begin{example}\label{example:nonterminating-transducer}
    Define the bitstring transducer $H$ to accept and output the binary representation of the string ``\texttt{nice}'', if and only if the input has the form $y \concat y \concat z$ for some finite bitstring $y \in \{0, 1\}^+$ and suffix $z \in \{0, 1\}^*$.
\end{example}
$H$ clearly satisfies \Cref{def:extensible-transducer}.
However, when inputs are produced by a uniform-like source of randomness (such as a PRNG), the probability of the input having the form $y \concat y \concat z$ becomes increasingly rare as the length of the input increases.

That is, when the input bitstring is sampled from $\mathcal{I}$, \Cref{example:nonterminating-transducer} has a strictly non-zero chance of never terminating.
\begin{proposition}\label{prop:h-is-invalid}
The transudcer $H$ defined in \Cref{example:nonterminating-transducer} has $\mu(A_H^\omega) < 1$.
\end{proposition}
\begin{proof}
Observe first that since $A_G^\omega$ is a union of disjoint sets $S_b$, we have:
$$\mu(A_G^\omega) = \sum_{b \in A_G} \mu(S_b)$$

We will compute an upper bound for $\mu(A_H^\omega)$.
$A_H$ contains bitstrings of length $2k, k \in \mathbb{N}^{>0}$ of the form $\mathit{xx}$,  for some $x \in \{0, 1\}^*$, and such that no prefix of $\mathit{xx}$ is of the form $\mathit{yy}, y \ in \{0, 1\}^*$.

Below we enumerate all members $b \in A_H$ with length $|b| = 2k$ for $k = 1, 2, 3, 4$, as well as the total probability of their bitstrings being sampled, $\sum_{|b| = 2k} \mu(S_b)$:
\begin{center}
\begin{tabular}{c|c|c}
    k & Bitstring               & Total Probability \\
    1 & \begin{tabular}{@{}c@{}}
         0 0 \\
         1 1
        \end{tabular}           & $\frac{2}{2^2}$ \\ \hline 
    2 & \begin{tabular}{@{}c@{}}
         01 01 \\
         10 10 
        \end{tabular}           & $\frac{2}{2^4}$ \\ \hline
    3 & \begin{tabular}{@{}c@{}}
         010 010 \\
         011 011 \\
         100 100 \\
         101 101
        \end{tabular}           & $\frac{4}{2^6}$ \\ \hline
    4 & \begin{tabular}{@{}c@{}}
         0110 0110 \\
         0111 0111 \\
         0100 0100 \\
         1011 1011 \\
         1000 1000 \\
         1001 1001
        \end{tabular}           & $\frac{6}{2^8}$ \\
\end{tabular}
\end{center}

Observe that for any $k \geq 2$, if $b \in A_H$ and $|b| = 2k$ then $b$ cannot have $00$ or $11$ as a prefix (and thus must have prefix either $01$ or $10$).
For $k \geq 2$, the number of bitstrings of length $2k$ within $A_H$ is at most $2 * 2^{k-2} = 2^{k-1}$.
Since there are $2^{2k}$ length $2k$ strings, we know $\sum_{|b|=2k} \mu(S_b) \leq \frac{2^{k-1}}{2^{2k}}$.

It follows that:
\begin{equation*}
\begin{aligned}
     & \mu(A_H^\omega) = \sum_{b \in A_H}{\mu(S_b)} = \sum_{k=1}^{\infty} \sum_{|b| = 2k}\mu(S_b) \leq \frac{2}{2^2} + \sum_{k=2}^{\infty} \frac{2^{k-1}}{2^{2k}} = \frac{1}{2} + \sum_{k=2}^{\infty} 2^{-(k+1)} = \frac{1}{2} + \frac{1}{4} = \frac{3}{4}
\end{aligned}
\end{equation*}

Hence, the probability of $B$ terminating is at most $\frac{3}{4} < 1$.

In fact, a more exact result exists in literature.
Let the number of ``square'' bitstrings without a square prefix of length $2k$ be the sequence $a_k$.
In our case, we have $a_k = 2^{2k} \sum_{|b| = 2k} \mu(S_{b})$.
In \citep{gabric2021borders}, the sum is computed as $\sum_{k=1}^{\infty} c_k = \sum_{k=1}^{\infty} a_k 2^{-2k} \approx 0.72996$, compared to our bound of 0.75.
\end{proof}

It is still possible to have \emph{practically} non-productive computations even if the transducer satisfies \Cref{def:extensible-transducer}: it is too weak, and thus the stronger condition of measure-theoretic \emph{eventual productivity} is needed.
We will see that the measure-theoretic condition used is indeed stronger than \Cref{def:extensible-transducer}:
\begin{proposition}
If $G$ is a generator, then $G$ is extensible.
\end{proposition}
\begin{proof}
Suppose the generator $G$ is not extensible.
Then there exists a bitstring $b$ where $G(b)$ rejects and for all strings $e \in \{0, 1\}^+$, $G(b \concat e)$ rejects, and thus $b \notin A_G$.
We know $\mu(S_b) > 0$.
Therefore $\mu(A^\omega_G) \leq 1 - \mu(S_b) < 1$, which is a contradiction.
\end{proof}

\section{Proof of \Cref{prop:rejection-sampling-valid}}\label{prop:termination-proof}
We will use the following definitions and lemmas in our proof.
\begin{definition}
For a generator $G$ and a bitstring $b \in \{0, 1\}^*$, define $s : \{0, 1\}^* \rightarrow \mathcal{P}(\{0, 1\}^*)$, the subset of bitstrings that concatenate with $b$ to reach $A_G$, as:
$$s(b) = \{ e \in \{0, 1\}^* : (b \concat e) \in A_G \}$$
\end{definition}

\begin{definition}
For generator $G$ and a bitstring $b$, define $e : \{0, 1\}^* \rightarrow \{X\} \cup \{0, 1\}^*$, the minimum extension, if one exists, such that $G(b \concat e(b))$ accepts as follows.
\begin{enumerate}
    \item If $G(b)$ accepts, then $e(b) = \epsilon$. (Empty string)
    \item If $G(b)$ rejects, then
        \begin{enumerate}
            \item If $s(b) \neq \emptyset$: $e(b) = \min(s(b))$, where the $\min$ is taken lexicographically.
            \item If $s(b) = \emptyset$: $e(b) = X$
        \end{enumerate}
\end{enumerate}
\end{definition}
Such an $f$ exists, since the set $s(b)$ is either empty, or bounded below (in lengths of elements) by 0, thus a lexicographically minimal element of $s(b)$ exists for all non-empty $s(b)$.
We only take the minimum when $s(b)$ is non-empty, and $e(b)$ is defined for all other cases.

\begin{lemma}\label{prop:termination-cond}
Let $G$ be a generator. 
If for all bitstrings $b \in \{0, 1\}^*$, $e(b) \neq X$ and there exists $c \in \mathbb{N}$ such that $|e(b)| \leq c$, then $G$ will terminate with probability 1.
\end{lemma}
\begin{proof}
For event $E$, let $\mathbb{P}(E | b)$ denote the probability of $E$ occurring, given that $b$ has been consumed on the input tape of $G$. 
Clear if $G(b)$ accepts, then $\mathbb{P}(G\ \text{accepts in}\ x\ \text{steps} | b) = 1$ for any $x$.
Thus assume $G(b)$ rejects.
Note that
$$\mathbb{P}(G\ \text{accepts in $|e(b)|$ steps} | b) \geq 2^{-|e(b)|}$$ 
where the inequality is due to the potential for multiple extensions of the same length as $e(b)$.
Thus 
$$\mathbb{P}(G\ \text{does not accept in $|e(b)|$ steps} | b) \leq 1 - 2^{-|e(b)|}$$

Suppose $G$ does not terminate on bitstream $\mathit{bs}$, we can partition $\mathit{bs}$ into the concatenation of a countable set of finite bitstrings, $\mathit{bs} = b_1...$, using $f$:
\begin{align*}
b_0 &= \epsilon    & \\
|b_i| &= |e(b_0b_1..b_{i-1})| & \qquad i > 1
\end{align*}

Define $x_i = b_1...b_i$, we can decompose the events that $G$ does not accept $\mathit{bs}$ as follows, and apply the earlier inequality to get:
\begin{align*}
      & \mathbb{P}(G\ \text{does not accept $\mathit{bs}$}) \\
 =    & \prod_{i=1}^{\infty} \mathbb{P}(|b_i| = |e(x_{i-1})| \cap G\ \text{does not accept on } b_i | x_{i-1}) \\
 \leq & \prod_{i=1}^{\infty} (1 - 2^{-|e(b_i)|})
\end{align*}

We see that if each $|e(b_i)|$ is bounded by some constant $c$, then the infinite product is guaranteed to collapse to zero.
\end{proof}

\begin{corollary}
The generator of natural numbers described in \Cref{def:naturalencoding} will terminate with probability 1.
\end{corollary}
\begin{proof}
Observe that for any sequence $b \in \{0, 1\}^+$ that does not terminate for the natural numbers generator, one of the two extensions $000$ or $00$ will result in the procedure terminating. 
Thus $|e(b)| \leq 3$ and \Cref{prop:termination-cond} applies.
\end{proof}

\begin{corollary}[Rejection Sampling is Sound]
Any procedure $F$ that performs rejection sampling based on predicate $f : \{0, 1\}^* \rightarrow \{T, F\}$ where there exists $w \in \{0, 1\}^*$ with $f(w) = T$ will terminate with probability 1.
\end{corollary}
\begin{proof}
The set of strings that $f$ returns $T$ for is fixed. 
Let $w$ be the shortest such string such that $f(w) = T$.
Thus the shortest extension that causes $F$ to terminate has size bound $|w|$, and \Cref{prop:termination-cond} applies.
\end{proof}

When designing generators, to ensure termination, it is enough to argue if there always exists a \emph{short} sequence of choices that leads to termination.
If the shortest sequence grows with execution time, then the generator has a chance of non-termination.

\ifarxiv
\section{Generator Acceptance and its Link To The Real Numbers}
\begin{observation}
For a generator $G$, consider the set of all finite bitstrings $x$ such that $G(x)$ accepts.
Let this set be $A$, then $A$ has the following properties:
\begin{enumerate}
    \item If $x \in A$, then for all $w \in \{0, 1\}^*, w \concat x \in A$.
    \item For all $y \in \{0, 1\}^*$, there exists a $w \in \{0, 1\}^*$ such that $y \concat w \in A$.
\end{enumerate}

Consider the map $f : \{0, 1\}^* \rightarrow [0, 1]$, defined as:
\begin{gather*}
    f(0 :: w) = 2^{-1}(0 + f(w)) \qquad f(1 :: w) = 2^{-1}(1 + f(w)) 
\end{gather*}
that is, $f(b) = 0.b$, where $b$ here is interpreted as digits in base-2.

Clearly $f$ is an injection. 
The set $A$ can thus be viewed as regions on the interval $[0, 1]$, through the mapping $f$.

$A$ is defined on finite sequences, however we can extend $A$ to be a set $A^+$ of infinite sequences by defining:
$$A^+ = \{ bs \in \{0, 1\}^\omega : \exists b. b \text{ is a finite prefix of } bs \cap b \in A \}$$

We can similarly define $f^+$ on $A^+$ by taking the fixed point of the computation $f$ (i.e.\ the result $f^+(bs)$ would be an infinite sum, and would equal the number $0.b \in [0, 1]$).
$f^+(A^+)$ has the nice property that the probability measure $\mu(f^+(A^+))$ is precisely the probability that $G$ eventually terminates.

We observe that $f(A)$ (and $f^+(A^+)$) has the following ``dense'' property: for any subinterval $(a, b) \subseteq [0, 1]$, $f(A) \cap (a, b) \neq \emptyset$, and furthermore contains infinitely many elements of $f(A)$.
\end{observation}
\fi

\section{Proof of \Cref{prop:re-equals-gen}}\label{proof:re-equals-gen}
The proof will be given in two parts, for the two sides of the inclusion.

\begin{proposition}\label{prop:re-subseteq-gen}
$\RE{} \subseteq \GEN{}$
\end{proposition}
\begin{proof}
Suppose $L \in \RE{}$.
We show that there exists a generator $G$ such that $\Lang{G} = L$.
Let $M$ be a Turing machine that accepts $x$ in a finite number of steps iff $x \in L$.
  
If $|L|$ is finite, then there exists a generator that produces one of the finite elements of $x \in L$ at random.
Thus we assume now $|L|$ is infinite.
We describe the generator $G$ for $L$ as follows:
\begin{enumerate}
    \item Compute some $x_0$ via dovetailing.
    \item Sample natural numbers $n$ and $m$ (e.g. via the procedure in \Cref{def:naturalencoding}).
    \item Sample a random string $x$ of length $m$, and simulate $M(x)$ for up to $n$ steps.
    \item If $M$ accepts during this simulation, then output $x$, otherwise output $x_0$.
\end{enumerate}

We now show that $G$ satisfies the conditions of a generator.
Because we are working in the case where $L \neq \emptyset$ (in fact $|L|$ is infinite), dovetailing is guaranteed to find some element $x_0 \in L$ in finite time.
As the simulation of $M(x)$ is for a bounded number of steps, the sampling procedure will always halt on any finite bitstrings.
For almost sure termination on bitstreams, note $G$ is deterministic subject to fixing the results of the two calls to \Cref{def:naturalencoding}.
Thus since \Cref{def:naturalencoding} terminates almost surely, so does $G$.

It clear that $G$ halts in an accepting state if and only if it has generated some $x \in L$, from which it follows that $\Lang{G} \subseteq L$.
For any $x \in L$ there exists $n$ such that $M$ accepts $x$ in $n_0$ steps.
Both $x$ and the corresponding $n > n_0$ has a set of finite bitsrings that results in them being selected in steps 2 and 3.
Thus $x$ and $n$ being selected occurs with nonzero probability.
We therefore have $\Lang{G} = L$ as required.
\end{proof}

\begin{proposition}
$\GEN{} \subseteq \RE{}$. 
\end{proposition}

\begin{proof}  
  Suppose $L \in \GEN{}$, and let $G$ be a generator such that $\Lang{G} = L$.

  We construct a Turing machine $M$ that terminates in an accepting state given input $x \in L$, and loops infinitely given input $x \notin L$, as follows.

  Given an input $x$, $M$ enumerates finite-length bitstrings in lexicographic order (0, 1, 00, 01, 10, etc.). For a bitstring $b$:

  \begin{itemize}

  \item $M$ simulates the execution of $G$ in an initial configuration where the input tape of $G$ contains $b$. By the definition of a generator (Definition~\ref{def:transducer}), this simulation of $G$ is guaranteed to lead to $G$ halting.

  \item If $G$ halts in an accepting state having generated $x$ (the input to $M$), $M$ halts in an accepting state.

  \item If instead $G$ halts in a non-accepting state or in an accepting state having generated some $x' \neq x$, $M$ proceeds to consider the next bitstring in lexicographical order.

    \end{itemize}

  By Definition~\ref{def:generated-language}, for $x \in \Lang{G}$ there exists a bitstring $b$ such that $G(b) = x$. Therefore, when executed on input $x$, $M$ will eventually encounter such a $b$, in which case $M$ will halt, accepting $x$.

  Conversely, for $x \notin \Lang{G}$ there is no bitstring $b$ such that $G(b) = x$. Therefore, when executed on input $x$, $M$ will loop forever, simulating the execution of $G$ on increasingly-long bitstrings.
\end{proof}

\section{Proof of \Cref{prop:subset-nondet}}\label{proof:subset-nondet}
\begin{proof}
We construct a non-deterministic TM $M$ for a generator $G$.
Give string $x$ with $|x| = n$, we decide if it belongs to $\Lang{G}$ on $M$ by simulating $G$ on $M$, branching nondeterministically whenever a random bit is revealed, with a branch for each of the possible outcomes.
When the simulation of $G$ terminates, accept if it is equal to $x$, and otherwise reject.
Thus, $x$ is accepted by $M$ if and only if the simulation of $G$ produced $x$ on any branch, that is, if and only if $x \in \Lang{G}$.
Since $G$ is time bounded by $t(n)$ / space bounded by $s(n)$, each branch of $M$ uses at most $log(n)t(n) + n \in O(t(n))$ time, and/or at most $s(n) \in O(s(n))$ space (assuming space is reused after writing each symbol of $x$).
Thus both claims hold.
\end{proof}

\section{Proof of \Cref{prop:cnf-sat-pg}}\label{proof:cnf-sat-pg}
\begin{proof}
We will define a generator $G_{m, k, p}$ for each $(m, k, p)$, where $m$ is the number of clauses, $k$ is the number of literals, $p$ is the length of the representation for each variable, with $m < k$.
Note all CNF-SAT formulas have $m < k$, as if the number of clauses is more than the number of literals ($m > k$), then some clauses must be empty leaving the formula unsatisfiable and therefore not a member of CNF-SAT.
Since we fix the number of clauses and literals, we know the exact length of the formula as long as the encoding of the clauses and literals are known.
Denote the encoding-dependent length of the formulas by $l(m, k, p)$, which is polynomial-time computable for any reasonable encodings.
For instance, $l(m, k, p) = m + (k + 2)p$ is the length of the encoding of the formula as follows: a list of $m$ clauses will contain $m$ length-1 separator characters; $(k+1)p$ characters will be consumed by $p$ literals of length $k$ each, with an extra character to encode negation; $k$ characters will be consumed by length-1 separators between literals.
Thus each string in $\Lang{G_{m, k, p}}$ have the same length, and the language of the generators $L_{m, k, p} = \Lang{G_{m, k, p}}$ forms a partition of the langauge CNF-SAT.

With these, we can define generators for CNF-SAT formulas with length exactly $n$ by first computing the set $I_n = \{ (m, k, p) : m + (k+2)p = n \}$, selecting a random $(m, k, p) \in I_n$ and running the corresponding $G_{m, k, p}$ to obtain a random formula of length $n$.

It now suffices to show that generators $G_{m, k, p}$ of satisfiable CNF formulas exist for any choice of $m, k, p$ with $m < k$, in time polynomial in $l(m, k, p)$. 
We define $G_{m, k, p}$ as follows:
\begin{enumerate}
    \item Choose an $n \leq k$ to be the number of variables. 
    \item For each variable $x_i, 1 \leq i \leq n$, choose a random boolean assignment. 
    \item Keep counters $m', k'$ initially set to $m, k$ respectively.
    \item Choose an integer $q \leq k'-m'$ to be the number of literals included in the clause, with the exception of $m'=1$, in which case $q = k'$. 
    Decrement counters $k', m' := k' - q, m' - 1$. 
    \item We form the literals $b_1, ..., b_q$ for the clause by selecting each literal to be either $x_j$ or $\neg x_j$ for some randomly $x_j$.
    Choose a number $1 \leq i \leq q$, and fix $b_i$ to be the $\textit{True}$ based on the preselected assignment.
    \item Add the clause $b_1 \cup ... \cup b_q$ to the formula $\phi$.
\end{enumerate}
Throughout the procedure, we only make random selections from finite sets.
Each choice of $q$ ensures that the remaining clauses will contain at least one literal.
The runtime is clearly polynomial in $l(m, k, p)$.
$\phi$ is satisfiable by construction as under the randomly chosen assignment during generation, each clause contains at least one $\textit{True}$ literal.

If $\phi$ is a satisfiable formula in CNF form with $m$ clauses, $k$ total literals and $n$ distinct variables, then \emph{some} selection of $G_{m, k, p}$ will generate $\phi$.
\end{proof}

\section{Proof of \Cref{prop:epg-subset-np}}\label{proof:epg-subset-np}
\begin{proof}
For $G = (G_i)_{i \in \mathbb{N}}$ generating $L$ in expected polynomial time, we need to construct a NDTM $M$ that decides $L$.
We define $M$ as follows. 
For input $x$ with $|x| = i$:
\begin{enumerate}
    \item Simulate $G_i$ on each branch for at most $p(i)$ steps.
    \item On each of the branches, branch again nondeterministically whenever a bit on the random input tape is read, for its possible contents, 0 or 1.
    \item When $G_i$ enters an accepting state, accept if the output tape contents is equal to $x$, otherwise reject.
\end{enumerate}

Without loss of generality, assume that $x \in L_i \subset L$.
Then we know there exists a bitstream $\textit{bs}$ where $x$ is accepted within $p(i)$ steps. 
Thus there would also be a branch on $M$ where $x$ is produced by the simulation of $G_i$ and thus $x$ would be accepted.

On the other hand, if $x \notin L$, then $x \notin L_i$ for any $i \in \mathbb{N}$, and no simulations of any $G_i$ would produce $x$. 
Thus all branches, and hence also $M$, would reject $x$.
\end{proof}

\section{Proof of \Cref{prop:pg-oracle-superset-np}}\label{proof:pg-oracle-superset-np}
\begin{proof}
$\PG{}^\NP{} \subseteq \NP{}^\NP{}$: follows since the $\NP{}^\NP{}$ machine can simulate the $\PG{}^\NP{}$ oracle machine nondeterministically.

$\NP{} \subseteq \PG{}^\NP{}$:
We use $\textit{SAT}$ as the \NP{} oracle.
For fixed $n \in \mathbb{N}$ and any $L \in \NP{}$ with NDTM $M$ time-bounded by $p$, we encode $M$ using the parametric $\textit{SAT}$ formula $\phi_n$, provided by the proof of \NP{}-completeness of $\textit{SAT}$~\citep{sipser1996introduction}.
$\phi_n$ has $n$ variables reserved for the initial contents of the tape, and for $x \in \{0, 1\}^{\leq n}, \phi_n(x)$ denotes the formula with those variables substituted for by the bits of $x$.

Note that if $|x| = n$, then by construction, $\phi_n(x)$ is satisfiable if and only if $M$ accepts $x$ within $p(n)$ steps.
And if $|x| < n$, then $\phi_n(x)$ is satisfiable if and only if there exists an extension $e$ with $|e| = n - |x|$, and $M(x \concat e)$ accepts.

Using the $\textit{SAT}$-solver, we describe a generator $G_n$ that samples a length-$n$ $x \in L$ as follows (if one exists):
\begin{enumerate}
    \item Suppose $s$ is the string already generated, and is initially the empty string.
    \item Run the $\textit{SAT}$ solver for both $\phi_n(s \concat 0)$ and $\phi_n(s \concat 1)$.
    \item If neither accepts, reject; if only $\phi_n(s \concat e)$ accepts, set $e' := e$; if both accepts, then choose at random one of $0$ or $1$ to be $e'$.
    \item Return to step 2 with string $s := s \concat e'$.
    \item Write $s$ to the output tape when $|s| = n$.
\end{enumerate}

Clearly this will terminate within $n$ iterations, and each iteration involves at most 2 calls to the $\textit{SAT}$ solver, for a total of $O(n)$ calls. 
By construction, for any $n$, any length-$n$ $x \in L$ can be generated in this fashion, thus $\Lang{G} = L$.
\end{proof}

\section{Proof of \Cref{prop:primes-epg}}\label{proof:primes-epg}
\begin{proof}
By Bertrand's Postulate (theorem), there is always a prime between $n$ and $2n$. 
Thus, we know that for any $n$, the set $\textit{PRIMES}_n = \{ x \in \textit{PRIMES} : 2^{|n|} \leq x < 2^{|n|+1} \}$ is nonempty.

We know from \citep{agrawal2004primes} that $\textit{PRIMES} \in P$.
Let its decider be a function $f : \{0, 1\}^* \rightarrow \mathbb{B}$ that runs in polynomial time $p(n)$.
For each $n \in \mathbb{N}$, define a generator $G_n$ as follows.
\begin{enumerate}
    \item Choose a random number $x$ between $[2^{n}, 2^{n+1} - 1]$ by sampling for $n$ bits.
    \item If $f(x)$ returns $\textit{Accept}$, output $x$.
    \item Otherwise, repeat from step 1.
\end{enumerate}

The density of primes in the range $[1, k]$ is $O(1 / log(k))$.
Assuming almost-uniform distribution of the primes, their density in the range $[2^{n}, 2^{n+1} - 1]$ is at least $O(1/log(2^{n+1})) = O(\frac{1}{n+1})$.
Therefore, an expected $O(n+1) = O(n)$ iterations is needed for the algorithm to find a prime number.
Since $f$ runs in polynomial time $p$, and we make $n+1$ expected calls to $f$, the expected runtime of the algorithm is $O(np(n))$, which is still polynomial.

Additionally, for every prime $q \in [2^n, 2^{n+1} - 1]$, there exists a bitstream $\textit{bs}$ which leads us to selecting $x = q$ in step 1 of the algorithm, and generate $x$ in time $O(p(n))$, lower than $O(np(n))$.
Thus $\textit{PRIMES} \in \EPG{}$.
\end{proof}

\section{Proof of \Cref{prop:boosting-epg}}\label{proof:boosting-epg}
\begin{proof}
Suppose $L$ is generated by the expected polynomial time generator $G = (G_i)_{i \in \mathbb{N}}$, with expected time of $p(i)$ for $G_i$ to produce a sample.

Let the runtime of $G_i$ be a random variable $X_i$ with $\mathbb{E}(X_i) = p(i)$.
By Markov's inequality, 
$\mathbb{P}(X_i > 3\mathbb{E}(X_i)) < \frac{\mathbb{E}(X_i)}{3\mathbb{E}(X_i)} = \frac{1}{3}$.
That is, the probability of that $X_i$ runs for $3\mathbb{E}(X_i) = 3p(i)$ steps without producing an output is at most $1/3$.
Therefore, $\mathbb{P}(X_i \leq 3p(i) \geq \frac{2}{3}$.

For each $G_i$, we define algorithms $H_i$ (these are operationally generators, but are not strictly generators since they have a nonzero probability of producing nothing---and generators require almost sure production of a value) to repeat for $k$ iterations the execution of $G_i$, running each execution for $3p(i)$ steps. 
$H_i$ therefore produces a sample of $L_i$ with probability at least $1 - \frac{1}{3}^k$, and runs in time $3kp(i)$.

To exceed probability $1 - \epsilon$, we have:
$
              1 - \frac{1}{3}^k > 1 - \epsilon 
    \implies \epsilon          > \frac{1}{3}^k 
    \implies log(\frac{1}{\epsilon}) < klog(3)
$.
Thus choosing $k = \lceil \frac{log(\frac{1}{\epsilon})}{log(3)} \rceil$ has the desired success probability.

It remains to show that $H_i$ can produce all elements $w \in L_i$ with nonzero probability. 
This holds because by definition, for all $w \in L_i$, there exists a bitstream $\textit{bs}$ that leads to $G_i$ producing $w$ in at most $p(i)$ steps.
This path can be executed by $H_i$, since $H_i$ simulates $G_i$ for $3p(i) > p(i)$ steps on each run, and will produce $w$ if $\textit{bs}$ is on the random tape when $G_i$ is being simulated.
\end{proof}

\section{Configuration Counting Argument for \Cref{prop:generator-time-space-tradeoff}}\label{appendix:config-counting-gspace-gtime}
Note that for \Cref{prop:generator-time-space-tradeoff}, an initial direct proof might be to proceed with a configurations-counting argument similar to the decision version, where the space-bounded machine is also automatically bounded since it must terminate.
However, due to the fact that generators take input, the situation is more complicated: the input, although read-only, allows the space-bounded machine to stall without exploring its state configuration space fully.
Whereas with the same inputs, the time-bounded machine must terminate in bounded time.
As a concrete example, consider a logspace generator $W$ which implements \emph{waiting} for the right bitstring input, as follows:
\begin{example}\label{example:logspace-waiting} 
Suppose the logspace generator $W = (W_i)_{i \in \mathbb{N}}$ implements the following for each $i$:
\begin{enumerate}
    \item \label{example:logspace-waiting-1} If the input tape header points to a $1$, move the input header forward, and the rest of the machine remains in the starting configuration.
    \item Otherwise, proceed with a logspace generation algorithm for $\Lang{W_i}$.
\end{enumerate}

The waiting process of step \ref{example:logspace-waiting-1} does not consume any space: nothing is written onto the work tape.
Hence $W$ consumes logspace overall.
Clearly, $W$ can take more than polynomial number of steps for generation, if given a sufficiently long prefix of $1$s on the input bitstring.
\end{example}

We give an alternative proof (similar to the proof of decision procedures) to highlight the intuition behind the the time-space tradeoff for generators, and its similarities/differences to the decision case.

\begin{proof}[Alternative Proof to \Cref{prop:generator-time-space-tradeoff}]
Suppose we are given a logspace generator $G = (G_i)_{i \in \mathbb{N}}$.
We describe a polynomial-time sampling procedure $F_i$ that works for each $\Lang{G_i}$.

\textbf{Definition of the DFA}: for each $G_i$, we produce (based on the transition function of $G_i$) a DFA $D_i$.
We first define configurations of a generator $H$:
\begin{definition}[Configuration of $H$]
A configuration $c$ of a generator $H$, is an element $c = (q, b_i, m, w_{work}, c_o) \in Q \times \{0, 1\} \times \{0, 1, \ldots, k\} \times \Sigma^* \times \Gamma$, where $q$ is the internal state of $H$, $b_i$ is the current symbol on the input tape, $m$ is the worktape head position, $w_{work}$ is the worktape content, $c_o$ is the current symbol on the output tape. 
The length of the worktape content is $k$.
\end{definition}
\noindent
Define $n = i$.
Note that as the work tape of $G_i$ is space bounded by $f(n)$ for a logarithmic function $f \in O(log(n))$, we only need to consider up to $f(n)$ elements on the worktape within any configuration of $G_i$. 
We define the nodes of $D_i$ to be all configurations of $G_i$, with up to $f(n)$ elements on the work tape.
That is, states of $D_i$ are from the set $Q \times \{0, 1\} \times \{0, 1, \ldots, k\} \times \Sigma^{f(n)} \times \Gamma$.
Define the starting state of $D_i$ to be $c_0$.
$c_0$ has two outward transitions, consuming symbols $0$ and $1$ respectively, heading towards the two possible starting configurations of $G_i$: the starting configurations with either a $0$ or a $1$ as the first symbol on the input tape. 
Define also the accepting state of $D_i$ to be $c_A$, a newly added configuration of $G_i$, with $\epsilon$ transitions $c^k_A \rightarrow c_A$ from all accepting configurations $c^k_A$ of $G_i$ (these are configurations with $q \in Q$ in the accepting state).

Define a transition on $D_i$ of $c \xrightarrow{1} c'$ (that consumes a $1$), if $c$ to $c'$ is a valid step based on $G_i$'s transition function, and $c'$ is identical to $c$, except the contents of the input cell is replaced by a 1. 
Define $c \xrightarrow{0} c'$ and $c \xrightarrow{\blanksymbol{}} c'$ in a similar way.
Define a transition $c \xrightarrow{\epsilon} c'$ for configuration changes that do not move the input tape forward.
The use of $\epsilon$ transitions are for conceptual convenience and do not make our machine an NFA: since $G_i$ is deterministic, at most one such $\epsilon$-transition is possible from each configuration, and that $\epsilon$-transitions cannot coexist with other types of transitions.
They can equivalently be replaced by some other reserved symbol, such as $\Box$, however such a formulation would be more verbose.

Note that since $|Q|, |\Sigma|, |\Gamma|$ are constants and $f$ is a logarithmic function, the total number of nodes on $D_i$ are bounded by $p(n)$, for some polynomial $p$.
There can be at most a polynomial number of edges, where a loose bound is the number of edges on a complete graph.
Thus, the total number of nodes and edges of $D_i$ is bounded by a polynomial $q(n)$.

\begin{proposition}
$D_i$ is a DFA.
\end{proposition}
\begin{proof}
The edges of $D_i$ are defined exactly according to the transition function of $G_i$.
Since the nodes of $D_i$ covers all accessible configurations of the TM (only a single cell for the input and output tapes are needed as the header cannot move back), each transition of $G_i$ that does not move the input header corresponds to at most one edge $c \xrightarrow{\epsilon} c'$ on $D_i$.
As $G_i$ is a deterministic TM (when fixing the input bitstring), there is at most one $\epsilon$ edge from any configuration.

For each transition that moves the input tape forward, there is also a unique next state once the revealed character $b$ is fixed. 
Thus, there cannot exist more than one of each kind of $0$, $1$, and $\blanksymbol{}$ edge, from each node on $D_i$.

Altogether, from each configuration there is at most one edge for each consumable symbol $\{0, 1, \blanksymbol{}\}$, making $D_i$ a DFA.
\end{proof}

\textbf{Polynomial Generation Algorithm}: for our \PG{} algorithm $F_i$, we describe a traversal of the DFA to find an accepting path, corresponding to an accepting execution of $G_i$.
We will use the following definitions:
\begin{definition}\label{def:write-transitions}
A transition $c \xrightarrow{ s } c'$ ($s \in \{0, 1, \epsilon\}$) is said to \emph{write} the symbol $x \in \Gamma$, if $c$ points to an empty cell on the output tape, and $c'$ points to a cell with $x$ on the output tape.
We also say such a transition is a \emph{writing transition}.
\end{definition}

\begin{definition}
For the execution path $p = c_1 \ldots c_k$ on $M_i$, the output of the path $p$ is the output symbol from all writing transitions used, in order of appearance in the execution path.
\end{definition}

Starting from $c_0$, perform a randomised BFS, without revisiting any nodes, until reaching the accepting state $c_A$.
This BFS results in a random, loopless path from $c_0$ to $c_A$.
The runtime of this algorithm is bounded above by $q(n)$, the total number of transitions and nodes in $M_i$, and thus runs in polynomial time of $n$.

\textbf{Correctness}: we will show the following: 1) any loopless path from $c_0$ to $c_A$ on $M_i$ corresponds to an input $b$ and a valid execution $G_i(b)$, and 2) any string generable by $G_i$ can be generated by $F_i$. 
That is, $\Lang{G_i} = \Lang{F_i}$. 

\begin{proposition}
Any path $p$ from $c_0$ to $c_A$ on $M_i$ corresponds to an input $b$ and an execution $G_i(b) = w$, which produces the same output $w$ as $p$.
\end{proposition}
\begin{proof}
Any path from $c_0$ to $c_A$, by definition, corresponds to an accepting execution on $G_i$.
Since $M_i$ is a DFA, the path $p$ corresponds uniquely to a string $t_1 \ldots t_k$ where $t_j \in \{0, 1, \epsilon \}$, representing the label of each transition that has been taken (since the path was accepting, by definition of $G_i$ there cannot be an blank symbol consumed).
Define $b$ to be the subsequence consisting of only symbols from the set $\{0, 1\}$.
Since $G_i$ is deterministic, and $M_i$ simulates the execution of $G_i$, $G_i(b)$ will visit precisely the configurations within $p$ during its execution, and produce the same output as $p$.
\end{proof}

\begin{proposition}
Any string generable by $G_i$ is the output of some \emph{loopless} path from $c_0$ to $c_A$ in $M_i$.
\end{proposition}
\begin{proof}
First, note that if the empty string is being generated, then any loops on the path can be removed to obtain another accepting path.
Thus assume now that the generated string is nonempty.

Suppose for a contradiction that a loop $r = c_1 \ldots c_k$ on $M_i$ exists within an accepting path $p$, with $c_1 = c_k$, and that it contains some writing transition $c_j \xrightarrow{s} c_{j+1}$ (assuming WLOG that $j \neq k$).

If $p = c_0 \ldots r \ldots c_A$, then repeating the loop $r$ within $p$ also yields an accepting path through $M_i$.
That is, $p_m = c_0 \ldots c_1 (r')^m \ldots c_A$, with $m \in \mathbb{N}$ are all accepting paths, where $r' = c_2 \ldots c_k$, and $(r')^m$ denotes the repetition of $r'$ m times.

Each $p_m$ also corresponds to an accepting execution on $G_i$.
Since there are writing transitions within $r$, and $G_i$ is write-once-only, there are at least $m$ characters on the output tape of $G_i$ in the corresponding execution of $p_m$.
In particular, we can choose $m > n$ to reach a contradiction, since $G_i$ is a generator that can only output strings of length $n = i$.
\end{proof}

Our algorithm $F_i$ samples a random loopless accepting path on $M_i$, and the above propositions together shows that every string from $\Lang{F_i}$ can be sampled this way.

With this, we can form our \PG{} generator $F = (F_i)_{i \in \mathbb{N}}$ to perform, on input $i$:
\begin{enumerate}
    \item Construct $M_i$ described above.
    \item Run the sampling procedure to generate a random string recognised by $M_i$.
\end{enumerate}
which has a nonzero chance of generating each instance of $\Lang{G_i}$ in polynomial time.
Thus $\Lang{G}$ is in \PG{}, as required.
\end{proof}

\section{Proof of \Cref{prop:lg-epg}}\label{proof:lg-epg}
\begin{proof}
First we make some necessary definitions.
\begin{definition}
For a relation $R \subseteq X \times Y$, define:
\begin{gather*}
Im(x) = \{ y \in Y : (x, y) \in R \} \qquad Dom(R) = \{ x \in X : \exists y. (x, y) \in R \}
\end{gather*}
\end{definition}

\begin{definition}
$M$ is a logspace transducer for the relation $R \subseteq X \times Y$ if, for all $x \in Dom(R)$, $M(x)$ is a generator for $\textit{Im}(x)$.
That is, $\Lang{M(x)} = Im(x)$.
\end{definition}

Specifically, logspace generators for $L$ are logspace transducers for their length partitions, $(L_i)_{i \in \mathbb{N}}$.
\begin{example}
A log-space generator $G = (G_i)_{i \in \mathbb{N}}$ is a nondeterministic logspace transducer for the following relation:
$$
R_G = \{ ((i, 1^{i}), x) \in C \times \Sigma^* : \exists b \in \{0, 1\}^*. G(i, b) = x \}
$$
\end{example}

\begin{proposition}[Uniform Sampling Of Logspace Transducers~\citep{arenas2019efficient}]
If the relation $R$ can be implemented by a logspace transducer, then there exists a polynomial-time sampler, such that on input $x$, outputs with success probability $> \frac{1}{2}$ a uniform sample from $\textit{Im}(x)$.
\end{proposition}

This leads to a straightforward \EPG{} algorithm for generating \LG{}, by applying this procedure repeatedly until a success is reached.
\end{proof}

\section{Proof of \Cref{prop:nspace-subseteq-exptime-generator}}\label{proof:nspace-subseteq-exptime-generator}
\begin{proof}
We will use two lemmas in the proof.
The first is the known result of time-space tradeoff for decision machines:
\begin{proposition}[Decision Time-Space Tradeoff]\label{prop:decision-time-space-tradeoff}
For space constructible $s: \mathbb{N} \rightarrow \mathbb{N}$ with $s(i) > log(i)$,
$$\NSPACE(s(n)) \subseteq \DTIME(2^{O(s(n))})$$
\end{proposition}

And the following:
\begin{lemma}\label{prop:partial-l}
For space constructible $s: \mathbb{N} \rightarrow \mathbb{N}$, and language $L \in \NSPACE(s(n))$, define the decision problem $\PARTIAL(L)$ as follows:
$$\PARTIAL(L) = \{ (x, 1^k) : \exists r. (x \concat{} r) \in L\ \land\ |r| = k \}$$
Then $\PARTIAL(L) \in \NSPACE(s(n))$
\end{lemma}
\begin{proof}
We give an NDTM $M$ with space bound $s$ that decides $\PARTIAL(L)$. 
Assume $L$ is decided by the NDTM $D$.
Given an input $(x, 1^k)$ with total length $n = |x| + k$, guess nondeterministically on each branch a string $r$ of length $k$.
Due to space constraints when $s(n) < n$, the string $r$ is not stored on the work tape but produced on-demand by nondeterministic branching for each symbol in $\Sigma$ when the input tape header moves forward.
We do not need to worry about the input tape header moving backwards and needing to ensure our choices are consistent, since this is assumed if $s(n) < n$. 

Then, simulate branches of $D$ on the string $x \concat{} r$, accepting if $D(x \concat{} r)$ accepts.
If there exists a $r \in \{0, 1\}^k$ such that $x \concat r \in L$ then there is a branch that chooses this $r$, and also there exists an accepting branch on the ensuing simulation of $D$ on $x \concat{} r$, thus $(x, 1^k)$ is accepted.
Conversely, if there doesn't exist such an $r$ then all branches on the ensuing simulation will reject for all choices of $r$, thus in this case $M$ will reject $(x, 1^k)$.
Hence $M$ decides $\PARTIAL(L)$.

As $D$ runs in $s(n)$ space, the simulation of $D$ on $M$ will also run in $s(n)$ space.
We have argued earlier that guessing the string $r$ do not require storing the guess onto the work tape, thus space usage remains $s(n)$ for $M$.
\end{proof}

Given an NDTM $M$ that runs in space $s(i)$ on each branch with language $L = L(M)$, we need to produce a generator $G = (G_i)_{i \in \mathbb{N}}$ such that for each $i$, $G_i$ runs in time $2^{O(s(n))}$ and has $\Lang{G_i} = L_i$, the strings of $L$ with length exactly $i$.

\Cref{prop:decision-time-space-tradeoff} and \Cref{prop:partial-l} together implies that $\PARTIAL(L)$ is decideable in time $O(2^{O(s(n))})$, thus size-$i$ instances of $\PARTIAL(L)$ can be decided on the generator $G_i$.
Let its decider be $P$.
$G_i$ proceeds as follows:
\begin{enumerate}
    \item Run $P(\epsilon, 1^i)$. 
        If this run rejects, then $G_i$ also rejects.
    \item Initialise $o = \epsilon$.
    \item Generate a random permutation of the $m$ symbols $c \in \Sigma$, $c_1, \ldots, c_k$. 
    \item For $c_i = c_1, \ldots, c_k$, run $P(o \concat c_i, 1^{i - |o|})$ until one of the $c_i$s accepts.
    Write this $c_i$ to the output tape, and apply the update $o := o \concat{} c_i$.
    \item If $|o| = i$, then terminate. 
    Otherwise continue from step 2.
\end{enumerate}
If $L_i = \emptyset$, then $P(\epsilon, 1^i)$ will reject on step 1, thus $\Lang{G_i} = \emptyset = L_i$.
Otherwise there must exist some choice of $s_1, \ldots, s_i$ at each step, such that $s = s_1 \concat{} \ldots \concat{} s_i \in L_i$.
It is safe to fix the $c_i$s by writing to the output tape since the run of $P$ ensures there exists an extension from that point onwards which belongs in $L_i$.
A different random permutation is necessary for each iteration to ensure all strings in $L_i$ have a chance of being sampled: if the same permutation is used for each step then only the first string from the permutation-induced lexicographical ordering will be produced by the algorithm.
Altogether we have $\Lang{G_i} = L_i$ and hence $\Lang{G} = L$

At most $|\Sigma|$ calls to $P$ are made on each iteration, together with the extra call of step 1.
Thus $i|\Sigma| + 1 \in O(i)$ calls to $P$ are made in total. 
Auxillary procedures take time $O(i)$.
Since $s(i) > log(i)$, the total runtime is $O(i2^{O(s(i))} + i) = O(2^{O(s(i))})$.
Therefore $L = \Lang{G} \in \GTIME(2^{O(s(n))})$.
\end{proof}

\section{Proof of \Cref{prop:reach-lg}}\label{proof:reach-lg}
\begin{proof}
The construction is based on certificate sampling and instance reconstruction (a la \Cref{prop:generator-verifier-equivalence}).
We first define the partition to work on.
Let $C = \{ (v, e) \in \mathbb{N} \times \mathbb{N} : e \leq v^2 \}$, where the tuple $(v, e) \in C \subseteq \{0, 1\}^*$ is interpreted as the binary encoding of two numbers, with $v$ representing the number of vertices and $e$ the number of edges. 
$l : C \rightarrow \mathbb{N}$ is defined as $l((v, e)) = \lceil log(v) \rceil + 2e* \lceil log(v) \rceil + 2 \lceil log(v) \rceil$.
$l((v, e))$ computes length of string needed to encode graph with $v$ vertices (encoded as a natural number), $e$ edges (encodes as an edge list), along with two nodes $i, j$ in the graph.
$l^{-1} : \{1\}^* \rightarrow \mathcal{P}(C)$, where $l^{-1}(1^n)$ is defined as a logspace machine that enumerates all possible pairs of integers $v, e$ in their binary representation, with $v, e \leq n$, and checks that $l((v, e)) = n$---output if true, otherwise continues onto the next candidate.
$l^{-1}$ runs in logspace of input $1^n$, since we are using binary representations.

We can iterate through the elements from the set $\{(v, e) : l(v, e) = n\}$ in logarithmic space, choosing a random element to stop at and run some $G_{v, e}$.
It suffices to show that each $G_{v, e}$ exists and runs in logarithmic space of $n = l(v, e)$.

We now describe a procedure that, when given $(v, e)$, generates a graph $G$ with $v$ vertices (labelled by natural numbers), $e$ edges, and two nodes $i, j$ within the graph.
The generated $G$ will, by construction, have a directed path between $i$ and $j$.

Denote $n = l((v, e))$.
The procedure is as follows:
\begin{enumerate}
    \item Output $v$ in unary, to denote the number of vertices in the graph.
    \item Select two random nodes $v_i$, $v_j$ with $i, j \leq v$.
    \item Generate a path of length $l \leq v$ via a random walk that starts from $v_i$ and ends at $v_j$:
    \begin{enumerate}
        \item On each step of the random walk, we write the edge taken to the output tape.
        \item When $v_j$ is reached, stop the random walk. 
        \item If, by the time we traversed $l-1$ nodes and still have not reached $v_j$, go directly to $v_j$ on the next step.
    \end{enumerate}
    \item \label{prop:reach-proof-erdos-renyi} Add $e$ random edges to the graph in Erdos-Renyi fashion. 
    Output each chosen edge immediately.
    \item Output the earlier $i, j$.
\end{enumerate}
The generated graph $G$ has a path from $v_i$ to $v_j$ by construction, and thus the full outputted instance is a member of $\textit{REACH}$.

For any given directed graph $H$ with $v$ vertices, $e$ edges, and a path from $x$ to $y$ there exists a bitstring that makes the following choices:
\begin{enumerate}
    \item Chooses exactly $i=x, j=y$ to be the start/endpoints of the random walk.
    \item Chooses exactly the path from $x$ to $y$ on $H$ in the random walk.
    \item Fills in exactly the rest of the edges of $H$ in step \ref{prop:reach-proof-erdos-renyi}.
\end{enumerate}
and thus $\langle H, x, y \rangle$ is generable by the above procedure.

Throughout the procedure, we store the binary representation of $v$, the chosen vertices $i, j$, the counter for the random walk up to $v$ in length, and the counter for the number of edges up to $v^2$.
All of which are representable in space $O(log(n))$.
\end{proof}

\section{Proof of \Cref{prop:gspace-subset-nspace-subset-exptime-oracle}}\label{proof:gspace-subset-nspace-subset-exptime-oracle}
\begin{proof}
$\GSPACE(s(n))^f \subseteq \NSPACE(s(n))^f$ follows from the same simulation argument as for the non-oracle case.

For $\NSPACE(s(n))^f \subseteq \GTIME(2^{O(s(n))})^f$, we need the following lemma for the time-space tradeoff on oracle machines:
\begin{lemma}\label{prop:oracle-nspace-subseteq-dtime}
For length-bounded $f : \Sigma^* \rightarrow \Sigma^*$,
$\NSPACE(s(n))^f \subseteq \DTIME(2^{O(s(n))})^f$
\end{lemma}
\begin{proof}
Take the nondeterministic oracle machine $M_f$, space-bounded by the polynomial $s$ and recognising some $L \in \NSPACE(s(n))^f$.
Suppose we are given an $x \in \Sigma^*$ with $|x| = n$.
Note that the query tape has size at most $s(n)$, and that the contents of the query tape fixes the contents of the answer tape.
Since the largest size query that can be made is of length $s(n)$, the answer tape can record at most $2^{O(s(n))}$ nonempty cells (as $f$ is length-bounded by $2^{O(n)}$), which is indexable by a pointer of size $O(s(n))$.
As the answer tape is read-only, the number of possible configurations on both the query and the answer tapes is $2^{O(s(n))}$.
The work tape and the input tape contribute another $2^{O(s(n))}$ configurations to the total state space, meaning the overall possible configurations for $M_f$ remain $2^{O(s(n))}$.
Querying edges on the configuration graph can be done by querying $M_f$'s definition.
Hence, a standard graph algorithm through the configuration graph can decide whether $x \in L$ in time $2^{O(s(n))}$.
\end{proof}

An identical construction to the proof of \Cref{prop:nspace-subseteq-exptime-generator} then follows: as the oracle variant of $PARTIAL(L)$ is still within $\NSPACE(s(n))^f$, the bit-by-bit construction can be used to construct all elements of $L \in \NSPACE(s(n))^f$ at random.
\end{proof}

\section{Proof of \Cref{prop:default-elem-implies-pg-membership}}\label{proof:default-elem-implies-pg-membership}
\begin{proof}
Clearly if $L \in \PG$ then $\DEFAULT(L) \in FP$ by running the generator $G $for $L$.

Suppose now $\DEFAULT(L) \in FP$.
With a computable default element as a fallback, we are free to sample the $\NP$ computation even though there is a chance of failure (i.e. finding a rejecting branch).
This is similar in idea to the proof of \Cref{prop:re-subseteq-gen}.
The following procedure $G = (G_i)_{i \in \mathbb{N}}$ is an algorithm for sampling $L$, given an $\NP$ machine $M$ that decides $L$ in polynomial time.
Assume $\DEFAULT(L)$ is implemented by the machine $D$ which, on input $1^i$, will output either $\epsilon$ or a string in $L_i$.
Define $G_i$ as:
\begin{enumerate}
    \item Sample a random string $x$ of length $i$.
    \item Simulate a single branch of $M(x)$.
    Output $x$ if $M(x)$ accepts.
    \item Otherwise, run $D(1^i)$. 
    If $D(1^i) = y \neq \epsilon$ then output $y$, otherwise reject.
\end{enumerate}
$G_i$ runs in polynomial time since both $M$ and $D$ are polynomial-time bounded.
Clearly all of $x \in L_i$ is samplable this way.
\end{proof}

\section{Proof of \Cref{prop:pspaceg-eq-pspace-oracle}}\label{proof:pspaceg-eq-pspace-oracle}
\begin{proof}
For $\PSPACEG{}^f \subseteq \PSPACE{}^f$, we know from \Cref{prop:gspace-subset-nspace-subset-exptime-oracle} that $\PSPACEG^f \subseteq \NPSPACE{}^f$.
We show that Savitch's theorem still holds for the oracle case: $\NPSPACE{}^f \subseteq \PSPACE{}^f$, which implies $\PSPACEG{}^f \subseteq \PSPACE{}^f$.
Suppose $M_f$ is an NDTM deciding some $L \in \NPSPACE(s(n))^f$, and space-bounded by the polynomial $s$.
Given any $x \in \Sigma^*$ with $|x| = n$, we want to decide whether $x \in L$. 
In the proof of \Cref{prop:oracle-nspace-subseteq-dtime} we see that the number of configurations on $M_f$ remains $2^{O(s(n))}$.
Thus, using the divide-and-conquer method detailed in the proof of Savitch's theorem, we get a polynomial $O(s^2(n))$ space deterministic algorithm for $L$, using the same oracle $f$, implying $L \in \PSPACE^f$.

$\PSPACE{}^f \subseteq \PSPACEG{}^f$: suppose $M_f$ decides a language $L \in \PSPACE{}^f$.
$G_f$ is defined as follows: on input $n$, choose a enumeration of the strings $x \in \Sigma^n$, and simulate the $\PSPACE^f$ machine on $x$. 
If $M_f(x)$ accepts then accept and output $x$ with probability $1/2$.
If all strings $x$ are rejected by $M_f$, then $G_f$ rejects.
If all $x$s are exhausted before we chose to accept, loop back to the beginning of the enumeration.
Clearly $x \in L$ has length $n$ if and only if $x$ has nonzero probability of being generated.
\end{proof}

\section{Proof of \Cref{prop:sufficient-pg-basis}}\label{proof:sufficient-pg-basis}
\begin{proof}
We show that $L$ is in $\PG{}$ by constructing a generator $G_L = (G_{L, i})_{i \in \mathbb{N}}$ for it.
Since $K \in \PG{}$, let $G_K = (G_{K, i})_{i \in \mathbb{N}}$ be the generator for $K$, where $G_{K, i}$ runs in time $p_K(i)$ for some polynomial $p_K$.
Also, let $G_s$ be the generator from the Efficient Extension condition for $s \in K$. Since $G_s$ is polynomial-time, it runs in time $p_E(|s|)$ for some polynomial $p_E$.

For any index $n \in \mathbb{N}$, we define $G_{L, n}$ as follows:
\begin{enumerate}
    \item Use $G_{K, n}$ to generate a sample $s \in K$. Note that since $(L_n)_{n \in \mathbb{N}}$ partitions $L$ by size, and $K \subseteq L$, $G_{K, n}$ generates elements in $K \cap L_i$. Thus $|s| = n$.
    \item Run the extension generator $G_s$ to obtain a string $s'$.
    \item Output $s'$.
\end{enumerate}

The total runtime is dominated by the runtime of $G_{K, i}$ and $G_s$. Since both are polynomial in $n$, the composed generator runs in polynomial time.
By the property of $G_s$, $s \preccurlyeq s'$. Since $s \in K \subseteq L$, by Upward Closure, $s' \in L$. Also $|s'| = |s| = n$, so $s' \in L_i$.
For any $w \in L_i$, since $K$ contains all minimal elements of $L$, there exists some $s \in K$ such that $s \preccurlyeq w$.
Since we restrict extensions to same-length strings (by the definition of $G_s$), we must have $|s| = |w| = n$, so $s \in K \cap L_i$.
$G_{K, i}$ generates $s$ with non-zero probability.
$G_s$ generates $w$ from $s$ with non-zero probability (since $w \in \{s' : s \preccurlyeq s' \land |s'| = |s|\}$).
Thus, $G_{L, i}$ generates $w$ with non-zero probability.

Therefore, $G_L$ is a polynomial-time generator for $L$, and $L \in \PG{}$.
\end{proof}

\section{Examples of Generators via Basis Extension}\label{example:basis-extension-examples}
\begin{example}[CNF-SAT revisited]
    We show that CNF-SAT is in \PG{} using \Cref{prop:sufficient-pg-basis}.
    We define the size of a formula by the number of variables $n$ and clauses $m$, using a fixed-length encoding.
    We define the partial order $\preccurlyeq$ such that 
    $\phi_1 \preccurlyeq \phi_2$ if $\phi_1$ and $\phi_2$ have the same number of clauses, 
    and each clause in $\phi_2$ contains the literals of the corresponding clause in $\phi_1$ 
    (i.e., $\phi_2$ is obtained by adding literals to $\phi_1$).
    \begin{itemize}
        \item \textbf{Upward Closure}: Adding literals to a disjunctive clause makes it strictly easier to satisfy. 
        Thus if $\phi_1$ is satisfiable, $\phi_2$ is also satisfiable.
        \item \textbf{Samplable Minimal Basis}: The minimal elements $K$ are satisfiable formulas where no literals 
        can be removed from clauses while maintaining the structure (i.e., unit clauses).
        This corresponds to consistent partial assignments.
        We can sample these by generating random consistent assignments for the $n$ variables, 
        and using the active literals as the unit clauses.
        \item \textbf{Efficient Extension}: Given a consistent assignment (formula of unit clauses), 
        we can extend it to a full formula of the target size by adding random literals to each clause.
    \end{itemize}
\end{example}

\begin{example}[Hamiltonian Path]
    Let HAMPATH be the set of graphs (encoded as adjacency matrices of size $n^2$) that contain a Hamiltonian path.
    We define $\preccurlyeq$ by subgraph inclusion: $G_1 \preccurlyeq G_2$ if $E(G_1) \subseteq E(G_2)$.
    \begin{itemize}
        \item \textbf{Upward Closure}: Adding edges to a graph preserves the existence of a Hamiltonian path.
        \item \textbf{Samplable Minimal Basis}: The minimal elements are graphs consisting of exactly one Hamiltonian path (and no other edges).
        These can be sampled efficiently by generating a random permutation of the vertices.
        \item \textbf{Efficient Extension}: Given a graph $G$ consisting of a Hamiltonian path, we can extend it to a random graph $G'$ of the same size (adjacency matrix representation) by setting additional entries in the adjacency matrix to 1 with some probability.
    \end{itemize}
\end{example}

\section{Proof of \Cref{prop:generator-verifier-equivalence}}\label{proof:generator-verifier-equivalence}
\begin{proof}
$(\ref{cert-verif}) \implies (\ref{cert-gen})$ is clear, as the composition of the certificate generator with the certificate recoverer leads to a polynomial-time generator of $L$.
That is, suppose the recovery scheme is $(S_V, R_V)$.
Then $G_i = R_i \circ S_i$ is the generator for $L_i$ that runs in time $p(i, q(i))$, which is polynomial in $i$.
And thus $L$ has a polynomial-time generator. 

To show $(\ref{cert-gen}) \implies (\ref{cert-verif})$, we need to construct a verifier $V$, along with its associated recovery scheme $(S_V, R_V)$, from a polynomial-time generator $G = (G_i)_{i \in \mathbb{N}}$ of $L$.

Define $V$ to use bitstrings as certificates. 
On input $w\#c$ with $|w| = n$, simulate $G_n$ on input $c$ until termination, and accept if the output of $G_n$ is $w$, reject otherwise.
Since $G_n$ will terminate in time $p(n)$, for a polynomial $p$, it will consume at most $p(n)$ random bits during any run. 
We define the certificates $c$ of $w$ to be any bitstream of length up to $p(n)$ that leads to $w$ being produced by $G_n$.
If $w \in L$ then there exists a bitstream (certificate) $c$ that leads to $V(w\#c)$ accepting. 
On the other hand, if $w \notin L$, then $V(w\#c)$ will always reject for any bitstrings $c$.
Thus all conditions of $V$ being a polynomial-time verifier for $L$ is satisfied.

We now need to show that the certificate recovery scheme $(S_V, R_V)$ exists.
Producing a random certificate is easy, as we are using bitstrings as certificates. 
By definition, $G_i$ terminates in time $p(i)$. 
For $S_V = (S_i)_{i \in \mathbb{N}}$, we define each $S_i$ as follows:
\begin{enumerate}
    \item Simulate $G_i$ to determine if $L_i$ is empty. Reject if $L_i = \emptyset$.
    \item Output a random bitstring of length $p(i)$.
\end{enumerate}

To see that $S_V$ is indeed a certificate generator, suppose firstly $w \in L_i = \Lang{G_i}$.
Then there exists some $p(i)$-length bitstring $c$ where $G_i(c) = w$, hence $V(w\#c)$ accepts by definition, and $w \in V_f(c)$.
We know $c$ is samplable from $S_i$ since it samples all bitstrings of length $p(i)$, hence $c \in \Lang{S_i}$.
Thus $w \in V_f(c)$, with $c \in \Lang{S_i}$.
If $L_i = \emptyset$ then $S_i$ indeed rejects.
Therefore we have $L_i \subseteq \Sigma^{i} \cap \bigcup_{c \in \Lang{S_i}} V_f(c)$.

By definition of $V$, if $V(w\#c)$ accepts then $w \in L$. 
Hence $V_f(c) \subseteq L$, and $\Sigma^{i} \cap V_f(c) \subseteq L_i$ for any $c$.
It follows $\Sigma^{i} \cap \bigcup_{c \in \Lang{S_i}} V_f(c) \subseteq L_i$.

Altogether, we have condition 1 of the certificate generator.
Condition 2 of the certificate generator clearly holds since $S_i$ takes exactly $p(i)$ steps, which is a polynomial in $i$.

For $R_V$, we define for each $i \in \mathbb{N}$, $R_i(c, b) = G_i(c)$.
That is, each $R_i$ will simulate $G_i$ on the certificate $c$ it received on the certificate input tape, and ignore its own random tape.
Condition 2 of certificate recoverer holds since $L_i \cap V_f(c)$ contains the single string $G_i(c)$, which is the only output of $R_i(c, b)$ for any choices of bitstring $b$.
Finally, condition 3 of certificate recoverer holds as $G_i$ accepts within $p(i)$ steps, which is a polynomial in parameters $i$ and $|c|$ in the trivial way.
Thus $R_V$ also satisfies all conditions of a certificate recoverer.
\end{proof}

\section{Proof of \Cref{prop:pg-closed-ops}}\label{proof:pg-closed-ops}
\begin{proof}
Suppose $G = (G_i)_{i \in \mathbb{N}}, H = (H_i)_{i \in \mathbb{N}}$ with $\Lang{G}, \Lang{H} \in \PG{}$.
Assume $G$ and $H$ run in polynomial time bounds $p$ and $q$ respectively.

\paragraph{Concatenation.}
Define $G \concat{} H$ as follows. 
For index $i$:
Sample a random $j \in 1, \ldots, i-1$ and define $k = i-j$.
Run $G_j$ first, followed by $H_k$ on an empty work tape immediately after $G_j$ terminates.
If any of $G_j$ or $H_k$ rejected, set $j := j+1$ and retry. 
If all of $j \in 1, \ldots i-1$ failed, then reject.
The total time is bounded by $i(p(i) + q(i))$ to generate an input of size $i$, which is still polynomial in $i$.

\paragraph{Kleene Star.}
For $G^*$, given an index $i$, we need to produce a length $i$ string that is the concatenation of shorter strings from $\Lang{G}$.
At first sight, this seems difficult, since the partition function for the integer $i$ grows at rate $O(2^{\sqrt{n}})$, and each partition needs to be checked to determine whether $\Lang{G^*_i} = \emptyset$.

However, we can define a polynomial time generator for $\Lang{G^*_i}$ by reducing to the Knapsack problem.
First run all $G_1, \ldots, G_i$ in sequence to obtain a sequence of $l_1, \ldots, l_m$ string lengths where for all $j \in 1, \ldots m$, $\Lang{G_{l_j}} \neq \emptyset$.
Then run the dynamic programming algorithm of Knapsack to extract a random way to sum up to $i$ using $l_1, \ldots, l_m$.
Reject if no solution exists.
Finally, run each $G_{l_j}$ in sequence, concatenating the outputs from each run to produce $G^*_i$.
This algorithm clearly produces some string in $\Lang{G^*_i}$ if one exists.
On the other hand, each partition $l_1, \ldots, l_m$ of $i$ has nonzero probability of being extracted from the Knapsack algorithm, and for each $j$, any string $x_j \in \Lang{G_{l_j}}$ has nonzero probability of being produced by their respective generators. 
Thus, all members of $\Lang{G*_i}$ also have nonzero probability of being sampled.

Determining the sequence of $l_j$s takes time $O(ip(i))$, Knapsack runs in time $O(mi) = O(i^2)$, and the final output stage takes time bounded above by $O(ip(i))$.
Thus the total runtime remains polynomial in $i$.

\paragraph{Union.}
For index $i$, select one of $G_i$ or $H_i$ at random to run first. 
If the first machine rejects, then run the second machine.
Reject if both runs reject, otherwise output what is written on the output tape from the first accepting run.
This generates all strings $x \in \Lang{G_i} \cup \Lang{H_i}$.
\end{proof}

\section{Proof of \Cref{prop:pg-not-closed-intersection}}\label{proof:pg-not-closed-intersection}
\begin{proof}
Suppose for a contradiction that given $L_1, L_2 \in \PG{}$, $L_1 \cap L_2 \in \PG{}$.
We will show that $\PG \supseteq \NP$, which contradicts \Cref{prop:crhf-implies-pg-neq-np}.
Take any $L \in \NP$, and define the following languages:
$$
    L^* = \{ 00x : x \in L \} \qquad A = \{ 01y : y \in \{0, 1\}^* \} \qquad B   = \{ 10y : y \in \{0, 1\}^* \}
$$
Observe that if $L^* \in \PG$, then also $L \in \PG$ by taking a suffix.

We claim both $L^* \cup A \in \PG{}$ and $L^* \cup B \in \PG$.
To see this, note that $\DEFAULT(L^* \cup A) \in FP$ as the function $f : \{1\}^* \rightarrow (L^* \cup A)$ defined as $f(\epsilon) = f(1) = f(11) = \epsilon$, and $f(1^k) = 010^{k-2}$ for $k \geq 2$ computes the default elements of $L^* \cup A$ in time linear in $k$.
By \Cref{prop:default-elem-implies-pg-membership}, $L^* \cup A \in \PG$.
$L^* \cup B \in \PG$ follows by a similar argument.

By assumption, $L^* = (L^* \cup A) \cap (L^* \cup B) \in \PG$, where the equality comes from the fact that $L^* \cap A = A \cap B = L^* \cap B = \emptyset$.
Thus, $L^* \in \PG$, which implies $\NP \subseteq \PG$, and we get the desired contradiction.
\end{proof}

\end{document}